\documentclass[aps,preprint]{revtex4}%
\usepackage{amsfonts}
\usepackage{amsmath}
\usepackage{amssymb}
\usepackage{graphicx}
\usepackage{hyperref}%
\newtheorem{theorem}{Theorem}

\newtheorem{proposition}[theorem]{Proposition}

\newenvironment{proof}[1][Proof]{\noindent\textbf{#1.} }{\ \rule{0.5em}{0.5em}}
\begin{document}
\title{BSW phenomenon for near-fine-tuned particles with external force: general
classification of scenarios}
\author{H.V.Ovcharenko}
\affiliation{Department of Physics, V.N.Karazin Kharkov National University, 61022 Kharkov, Ukraine}
\affiliation{Institute of Theoretical Physics, Faculty of Mathematics and Physics, Charles
University, Prague, V Holesovickach 2, 180 00 Praha 8, Czech Republic}
\author{O.B.Zaslavskii}
\affiliation{Department of Physics and Technology, Kharkov V.N. Karazin National
University, 4 Svoboda Square, Kharkov 61022, Ukraine}

\pacs{PACS number}

\begin{abstract}
If two particles moving towards a black hole collide in the vicinity of the
horizon, the energy $E_{c.m.}$ in the center of mass frame can grow
indefinitely if one of particles is fine-tuned. This is the Ba\~{n}ados, Silk
and West (BSW) effect. One of objections against this effect consists in that
for some types of a horizon fine-tuned particles cannot reach the horizon.
However, this difficulty can be overcome if instead of exact fine-tuning, one
of particle is nearly fine-tuned, with the value of small detuning being
adjusted to the distance to the horizon. Such particles are called
near-fine-tuned. We give classification of such particles and describe
possible high energy scenarios of collision in which they participate. We
analyze the ranges of possible motion for each type of particle and determine
under which condition such particles can reach the horizon. We analyze
collision energy $E_{c.m.}\,$and determine under which conditions it may grow
indefinitely. We also include into consideration the forces acting on
particles and find when the BSW effect with nearly-fine-tuned particles is
possible with finite forces. We demonstrate that the BSW effect with particles
under discussion is consistent with the principle of kinematic censorship.
According to this principle, $E_{c.m.}\,$cannot be literally infinite in any
event of collision (if no singularity is present), although it can be made as
large as one likes.

\end{abstract}
\maketitle
\tableofcontents

\section{Introduction}

At present, high energy particle collisions near black \ holes (and, more
generally, collisions in a strong gravitational field) remain hot topic. This
is mainly due to findings of Ba\~{n}ados, Silk and West who noticed that if
two particle move towards an extremal black hole and collide near its horizon,
under certain conditions the energy in the center of mass frame $E_{c.m.}$
becomes unbounded \cite{ban}. This happens if one of colliding particles has
to be critical (meaning that its radial velocity has to vanish on the
horizon). This is what is called the BSW effect. It also revived interest to
earlier works on this subject \cite{pir1} - \cite{pir3}. The aforementioned
conditions imply that one of colliding particles has fine-tuned parameters
(say, special relation between the energy and angular momentum).

A number of objections was pushed forward against this effect \cite{berti},
\cite{ted}, \cite{hod}, \cite{lib}. Their meaning can be reduced to the
statements that there are some factors that bound the energy of collision
$E_{c.m.}$ However, now it is clear that such a kind of objections does not
abolish the BSW effect. Moreover, according to the principle of kinematic
censorship, it is impossible to have literally infinite $E_{c.m.}$ in each
event of collision. Instead, this quantity remains finite but can be made as
large as one likes \cite{cens}. Therefore, the aforementioned objections
simply put limits of validity of the BSW effect but do not abolish it.

Moreover, it turned out that the requirement of having an extremal horizon is
not necessary for the BSW effect. One of factors that prevents infinite
$E_{c.m.}$ is correlation between the type of a trajectory of a fine-tuned
particle and a type of a horizon. For example, the critical particle cannot
reach the horizon of a nonextremal black hole. However, if one somewhat
relaxes the condition of criticality and replaces it by near-criticality with
certain relationship between detuning\ and proximity to the horizon, the\ BSW
effect becomes possible \cite{gr-pav}.

Another objection against the BSW effect is related to backreaction of
radiation on a particle. Meanwhile, it was shown that the BSW effect survives
under the action of a force for extremal \cite{tz13} and nonextremal horizons
\cite{tz14}.

A large number of particular results for the BSW effect invokes necessity of
constructing the most general scheme that would encompass all possible cases.
In the previous paper \cite{ov-zas-prd}, such a scheme with full
classification of scenarios leading\ to the BSW effect, was developed for
collisions in which fine-tuned particles participate. In the present paper, we
developed a corresponding scheme for near-fine-tuned particles thus
essentially generalizing the observation made in \cite{gr-pav}. In doing so,
we also take into account a force acting on particles, so in general they are
not free-falling. An important reservation to be mentioned is that we consider
motion within the equatorial plane of rotating axially symmetric black holes.
Another reservation consists in that we work in the test particle
approximation neglecting backreaction of particles on the metric.

The paper is organized as follows. In Section II we give general setup for
motion of particles in axially symmetric spacetimes. In Section III we
introduce different types of near-fine-tuned particle and analyze possible
ranges of their motion. In Section IV we focus on kinematical properties of
near-fine-tuned particles for different ranges of their motion that becomes
important in Section V in which we investigate behavior of energy in the
center of mass frame of two colliding particles. In section VI we give general
expressions for an acceleration experienced by near-fine-tuned particles. This
becomes useful in Sections VII, VIII and IX where we analyze near-horizon
behavior of acceleration for different ranges of motion of near-fine-tuned
particles. In the Section X we briefly formulate the results we obtained in
previous sections. In Section XI we check the validity of the aforementioned
principle of kinematic censorship. Section XII is devoted to a possibility of
varying of ranges of particle motion under the action of the external force.
In Section XIII we summarize corresponding results of our work.

\section{General setup\label{set}}

In this work, we are going to analyze the properties of the BSW phenomenon for
near-fine-tuned particles. At first, we need to define what we mean by this
term. We are investigating the motion of particles in the background of a
rotating black hole which is described in the generalized Boyer-Lindquist
coordinates $(t,r,\theta,\varphi)$ by the metric:%
\begin{equation}
ds^{2}=-N^{2}dt^{2}+g_{\varphi\varphi}(dt-\omega d\varphi)^{2}+\frac{dr^{2}%
}{A}+g_{\theta\theta}d\theta^{2},
\end{equation}
where all metric coefficients do not depend on $t$ and $\varphi$. The horizon
is located at $r=r_{h}$ where $A(r_{h})=N(r_{h})=0$. Near the horizon, we
utilize a general expansion for the functions $N^{2}$, $A$ and $\omega$:%
\begin{equation}
N^{2}=\kappa_{p}v^{p}+o(v^{p}),\text{ \ \ }A=A_{q}v^{q}+o(v^{q}),
\label{an_exp}%
\end{equation}%
\begin{equation}
\omega=\omega_{H}+\omega_{k}v^{k}+o(v^{k}), \label{om_exp}%
\end{equation}
where $q,p$ and $k$ are numbers that characterize the rate of a change of the
metric functions near the horizon, and $v=r-r_{h}.$

Now, let us investigate the motion of a particle in such a space-time. If a
particle is freely moving, the space-time symmetries with respect to
$\partial_{t}$ and $\partial_{\varphi}$ impose conservation of the
corresponding components of the four-momentum: $mu_{t}=-E$, $mu_{\varphi}=L$.
We assume the symmetry with respect to the equatorial plane. In what follows,
we restrict ourselves by equatorial motion. This allows us to write the
4-velocity of a free-falling particle in the following form:%
\begin{equation}
u^{\mu}=\left(  \frac{\mathcal{X}}{N^{2}},\sigma\frac{\sqrt{A}}{N}%
P,0,\frac{\mathcal{L}}{g_{\varphi\varphi}}+\frac{\omega\mathcal{X}}{N^{2}%
}\right)  , \label{4_vel}%
\end{equation}
where $\sigma=\pm1$, $\mathcal{X}=\epsilon-\omega\mathcal{L}$, $\epsilon=E/m$,
$\mathcal{L}=L/m$ and $P$ is given by:%
\begin{equation}
P=\sqrt{\mathcal{X}^{2}-N^{2}\left(  1+\frac{\mathcal{L}^{2}}{g_{\varphi
\varphi}}\right)  }. \label{P}%
\end{equation}

Now, let a particle move non-freely. In the case of an external force acting
on the particles, the quantities $\varepsilon$ and $\mathcal{L}$ are obviously
not conserved. However, despite this fact, we can still use the expression
(\ref{4_vel}) but with general functions $\mathcal{X(}r\mathcal{)}$ and
$\mathcal{L(}r\mathcal{)}$, provided forces do not depend on time. Near the
horizon, we can use the Taylor expansion for them:%
\begin{equation}
\mathcal{X}=X_{s}v^{s}+o(v^{s}),\text{ \ \ }\mathcal{L=}L_{H}+L_{b}%
v^{b}+o(v^{b}). \label{L_exp}%
\end{equation}

In further analysis, we will also require expressions for the tetrad
components of the 4-velocity. To this end, let us introduce the corresponding
tetrad:%
\begin{align}
e_{\mu}^{(0)}  &  =N\left(  1,0,0,0\right)  \},\text{ \ \ }e_{\mu}%
^{(1)}=\left(  0,\frac{1}{\sqrt{A}},0,0\right)  ,\label{tetr_1}\\
e_{\mu}^{(2)}  &  =\sqrt{g_{\theta\theta}}\left(  0,0,1,0\right)  ,\text{
\ \ }e_{\mu}^{(3)}=\sqrt{g_{\varphi\varphi}}\left(  -\omega,0,0,1\right)  .
\label{tetr_2}%
\end{align}

The tetrad components of 4-velocity read%
\begin{equation}
u^{(a)}=\left(  \frac{\mathcal{X}}{N},\frac{\sigma}{N}P,0,\frac{\mathcal{L}%
}{\sqrt{g_{\varphi\varphi}}}\right)  .
\end{equation}

\section{Classification of different near-fine-tuned
particles\label{sec_class_of_part}}

In this section, we are going to define and classify different near-fine-tuned
particles that will be useful for further analysis. The definition of
near-fine-tuned particles generalizing that introduced in previous works (see
\cite{gr-pav}) is as follows: a particle is called near-fine-tuned if
$\mathcal{X}$ near the horizon has the Taylor expansion:%
\begin{equation}
\mathcal{X=}\delta+X_{s}v^{s}+o(v^{s}),\text{ \ \ }s>0, \label{X_us_2}%
\end{equation}
where $\delta\ll1$ is a dimensionless parameter. It is important to note that
a particle for which the expansion for $\mathcal{X}$ starts with a constant is
generally called usual (and we will see several analogies with usual particles
in our further analysis). However, we are going to show that the case
$\delta\ll1$ requires a distinct analysis.

Before we proceed further, we have to note that in our analysis we will also
require that the time coordinate during the motion of the particle has to
increase (this is the so-called forward-in-time condition). To this end, the
time component of the 4-velocity has to be positive: $u^{t}>0$. From the
expression (\ref{4_vel}), one sees that this requires $\mathcal{X}\geq0$. In
further analysis, we will require this condition to hold for all particles
under consideration.

To analyze the behavior of the four-velocity, let us introduce a
classification of different types of particles based on different values of
$s$ (see Table \ref{class_tab}).

\begin{table}[ptb]%
\begin{tabular}
[c]{|c|c|c|}\hline
Condition & Type of particle with non-zero $\delta\ll1$ & Abbreviation\\\hline
$s<p/2$ & Near-subcritical & NSC\\\hline
$s=p/2$ & Near-critical & NC\\\hline
$s=p/2$ and (\ref{ultra_crit_prop}) & Near-ultracritical & NUC\\\hline
$s>p/2$ & Near-overcritical & NOC\\\hline
\end{tabular}
\caption{ Table showing classification of different near-fine-tuned particles}%
\label{class_tab}%
\end{table}

In further analysis, we will use abbreviations: "NSC" for near-subcritical,
"NC" for near-critical, "NUC" for near-ultracritical, and "NOC" for
near-overcritical particles.

As we will show, these particles will correspond to generalization, for
non-zero $\delta$, of subcritical, critical, and ultracritical particles
introduced in Ref. \cite{ov-zas-prd}. The only new type of particles is the
near-overcritical one that does not have any analog for $\delta=0$. But before
doing this, let us consider an expression for the radial component of the
four-velocity. As we will show, it depends strongly on a type of particle that
will justify the necessity of introduced clasiification. For this purpose and
for further investigation, let us consider the quantity $P$ near the horizon.
First of all, we note that exactly on the horizon, $P=|\delta|$ (this becomes
obvious if one substitutes equations (\ref{X_us_2}) and (\ref{an_exp}) into
(\ref{P}) and takes the limit $v\rightarrow0$). However, as one moves away
from the horizon, the quantity $P$ starts to differ from the value
$P=|\delta|$. Depending on the parameters of a particle, $P$ may either
decrease or increase. In the first case, at some radial distance $v_{t}$,
where the index "t" stands for "turning point", $P$ becomes zero: $P(v_{t})=0$
(hereafter, "distance" means "coordinate distance"). In the second case, there
are no turning points. However, we can still define effective coordinate
distances $v_{e}$, where the index "e" stands for "effective", at which $P$
changes by values comparable to the value of $P$ on the horizon. Formally, we
can define effective distances to be such that $P(v_{e})-P(0)\sim\delta$. As
we will show in further analysis, the physical properties of the collisional
process depend strongly on the point at which this process takes place and the
relationship between it and $v_{t}$ (or $v_{e}$).

Now, let us find under which conditions $P$ imposes the existence of roots and
in what regions a particle may move. Using the expression for $P$ (see
(\ref{P})), we have:%
\begin{equation}
P=\sqrt{\mathcal{X}^{2}-N^{2}\left(  1+\frac{\mathcal{L}^{2}}{g_{\varphi
\varphi}}\right)  }. \label{P2}%
\end{equation}

Reality of this expression is defined by condition
\begin{equation}
\mathcal{X}^{2}\geq N^{2}\left(  1+\frac{\mathcal{L}^{2}}{g_{\varphi\varphi}%
}\right)  . \label{p=0_eqq}%
\end{equation}

As we impose forward-in-time condition, $\mathcal{X}>0$. Also, by definition
$N^{2}$ is non-negative. Thus we can formally take "square root" of
(\ref{p=0_eqq}) and get:%
\begin{equation}
\mathcal{X\geq}N\sqrt{1+\frac{\mathcal{L}^{2}}{g_{\varphi\varphi}}}.
\end{equation}

Our task is to find such ranges of radial coordinate in which this condition
holds. To do this, let us at first solve equation $\mathcal{X=}N\sqrt
{1+\frac{\mathcal{L}^{2}}{g_{\varphi\varphi}}}.$ Substituting (\ref{X_us_2})
and (\ref{an_exp}) one gets
\begin{equation}
\delta+X_{s}v_{t}^{s}+o(v_{t}^{s})=\sqrt{\kappa_{p}\left(  1+\frac
{\mathcal{L}^{2}}{g_{\varphi\varphi}}\right)  }v_{t}^{p/2}+o(v_{t}^{p/2}).
\label{p=0_eqqq}%
\end{equation}
\qquad\qquad

We cannot find general expressions for all roots, but we can determine if
there are any new roots that were absent in the case of zero $\delta$. To
analyze the corresponding solutions, we consider different types of particles separately.

\subsection{Near-subcritical particles ($s<p/2$).}

In this case, it is obvious that the term $v^{s}$ is dominant over the term
$v^{p/2}$ that gives us $P\approx|\mathcal{X}|$. Therefore, if the
forward-in-time condition is satisfied, the quantity $P$ is real. Thus, the
entire analysis of regions of motion for near-subcritical particles is
restricted to an analysis of the positivity of $\mathcal{X}$. To analyze
regions where $\mathcal{X}\geq0$, we first need to find where $\mathcal{X=}0$.
By substituting (\ref{X_us_2}), we have:
\begin{equation}
\delta+X_{s}v_{t}^{s}=o(v_{t}^{s}),
\end{equation}
that gives
\begin{equation}
v_{t}=\left(  -\frac{\delta}{X_{s}}\right)  ^{1/s}+o(\delta^{1/s}).
\label{vt_subcr}%
\end{equation}

Note that this solution is only possible if $\delta>0$ and $X_{s}<0$ or if
$\delta<0$ and $X_{s}>0$. In the first case, it can be easily seen that
$\mathcal{X}$ is non-negative only in the range $[0,v_{t}]$, thus the particle
can only move from the horizon to $v_{t}$, and this region is not connected to
infinity. On the other hand, in the second case, $\mathcal{X}$ is non-negative
only in the range $[v_{t},\infty]$. In this section, when we write $\infty$,
it means that we cannot find other roots that limit the motion of the
particle. This implies that there are no roots generated by a non-zero
$\delta$ that restricts the motion of a particle (although some of them may
exist at distances greater than $v_{t}$, their existence is not defined by
$\delta$).

However, if the conditions $\delta>0$ and $X_{s}<0$, or $\delta<0$ and
$X_{s}>0$, do not hold, then there is no $\delta$-related root and a particle
may move in the range $[0,\infty]$. The forward-in-time condition for the
absence of roots holds only in the case when $\delta>0$ and $X_{s}>0$ (while
for the case $\delta<0$ and $X_{s}<0$, it does not hold).

As mentioned above, in the case of the absence of roots, we can only define
effective distances at which $P$ changes on the order of $\delta$. For
near-subcritical particles, where $s<p/2$, we have $P\approx\mathcal{X}%
\approx\delta+X_{s}v^{s}$. By comparing the two terms in the expansion of $P$,
we get:%
\begin{equation}
v_{e}\sim\delta^{1/s}. \label{vc_subcr}%
\end{equation}

To summarize, we have:%
\begin{equation}
\text{If particle is NSC (}s<p/2\text{), }v_{c}\in\text{ }\left\{  \text{%
\begin{tabular}
[c]{l}%
$\lbrack0,v_{t}]$ where $v_{t}=\left(  -\frac{\delta}{X_{s}}\right)  ^{1/s}$
if $\delta>0$ and $X_{s}<0,$\\
$\lbrack v_{t},\infty]$ where $v_{t}=\left(  -\frac{\delta}{X_{s}}\right)
^{1/s}$ if $\delta<0$ and $X_{s}>0,$\\
$\lbrack0,\infty]$ if $\delta>0$ and $X_{s}>0,$ in this case $v_{e}\sim
\delta^{1/s}.$%
\end{tabular}
}\right.  \label{fit_cond}%
\end{equation}

In this case, the ranges of motion were limited only by the forward-in-time
condition. Therefore, (\ref{fit_cond}) describes all cases when the
forward-in-time condition holds, for general expressions for $P$. This fact
will be used in further analysis.

We also have to make some reservation about ranges of motion of particles. Up
to now, we have not considered directions of their motion: ranges where
$\mathcal{X}>0$ and $P$ is real do not depend on whether particle is ingoing
or outgoing. However, it is obvious that in application to the real problem
this becomes important. If particle has the finite proper time to achieve
horizon from some finite distance, then its motion cannot be reversed, and, in
application to our problem, this describes only ingoing particle. (It can be
outgoing if one considers a white hole, but we will not focus on this
possibility). If the proper time diverges near horizon, a particle may move
not only towards a horizon but also from its small vicinity in the outward
direction, so our analysis is applicable both to ingoing and outgoing
particles. Analysis of the proper time was done in Appendix
\ref{sec_prop_time}, where we analyze behaviour of proper time for different
types of particles. One can see that for different types of particles the
behaviour of the proper time is different. Meanwhile, we still can conclude
that if $q<2,$ any particle has finite proper time, and in this case our
current analysis is applicable only for ingoing particles.

\subsection{Near-critical particles ($s=p/2$).}

In this case $v^{s}$ and $v^{p/2}$ terms are of the same order and eq.
(\ref{p=0_eqqq}) simplifies (we denote $A_{p/2}=\sqrt{\kappa_{p}\left(
1+\frac{\mathcal{L}^{2}}{g_{\varphi\varphi}}\right)  }$):%
\begin{equation}
\delta\approx\left(  A_{p/2}-X_{p/2}\right)  v_{t}^{p/2}.
\end{equation}

Solving this, we have:%
\begin{equation}
v_{t}\approx\left(  \frac{\delta}{A_{p/2}-X_{p/2}}\right)  ^{2/p}.
\label{vt_cr}%
\end{equation}

This solution is possible only in two cases: if $\delta>0$ and $A_{p/2}%
>X_{p/2}$, or if $\delta<0$ and $A_{p/2}<X_{p/2}$. In the first case, $P$ is
positive in the range $[0,v_{t}]$, while in the second case it is positive in
the range $[v_{t},\infty]$. In other cases, there is no solution and the
particle may move in the range $[0,\infty]$. In this case, we can only define
effective distances at which $P$ changes on the values of the order of
$\delta$. To find these distances, we write the expression for $P$ in the case
of NC particles.%
\begin{equation}
P\approx\sqrt{\delta^{2}+2\delta X_{p/2}v^{p/2}+\left(  X_{p/2}^{2}%
-A_{p/2}^{2}\right)  v^{p}}.
\end{equation}

Generally, we need to determine at which distances each term in the expression
for $P$ is compatible with $\delta^{2}$ and choose the dominant solution in
$\delta$. We will not provide a general analysis and will just present a
result: the effective distance $v_{e}$ is such that all terms are of the same
order. In this case,%
\begin{equation}
v_{e}\sim\delta^{2/p}. \label{vc_cr}%
\end{equation}

We need to determine, whether the forward-in-time condition holds in regions
of motion for the NC particle. First, let us consider the case where
$\delta>0$ and $A_{p/2}>X_{p/2}$. Then, the root (\ref{vt_cr}) is closer to
the horizon than (\ref{vt_subcr}) (this is true because $A_{p/2}>0$).
Therefore, in the entire range $[0,v_{t}]$, the forward-in-time condition holds.

Next, let us consider the case where $\delta<0$ and $A_{p/2}<X_{p/2}$. In this
case, the root (\ref{vt_cr}) is further from the horizon than (\ref{vt_subcr}%
). Therefore, in the entire range $[v_{t},\infty]$, the forward-in-time
condition holds.

The remaining cases are (i) $\delta>0$ and $A_{p/2}<X_{p/2}$ or (ii)
$\delta<0$ and $A_{p/2}>X_{p/2}$. In case (i), the forward-in-time condition
holds according to the third case in (\ref{fit_cond}). In case (ii), it only
holds for $X_{p/2}>0$ and for $v\in\lbrack(-\frac{\delta}{X_{p/2}}%
)^{2/p},\infty]$. In all other cases, motion is forbidden.

In summary, generalizing all the above facts, we have:%
\begin{equation}
\text{If particle is NC (}s=p/2\text{), }v_{c}\in\left\{  \text{
\begin{tabular}
[c]{l}%
$\lbrack0,v_{t}]$ where $v_{t}$ is in (\ref{vt_cr}) and $\delta>0$ and
$A_{p/2}>X_{p/2},$\\
$\lbrack v_{t},\infty]$ where $v_{t}$ is in (\ref{vt_cr}) and $\delta<0$ and
$A_{p/2}<X_{p/2},$\\
$\lbrack0,\infty]$ if $\delta>0$ and $A_{p/2}<X_{p/2},$ in this case
$v_{e}\sim\delta^{2/p},$\\
$\lbrack(-\frac{\delta}{X_{p/2}})^{2/p},0]$ if $\delta<0$ and $A_{p/2}%
>X_{p/2}>0,$ $v_{e}\sim\delta^{2/p}.$%
\end{tabular}
}\right.  \label{vt_nc}%
\end{equation}

\subsection{Near-ultracritical particles ($s=p/2$ and special condition)}

It may appear that coefficients in expansions for $P$ may be such that several
first terms in it cancel. This happens in a case when%
\begin{align}
\left(  \mathcal{X-}\delta\right)  ^{2}-N^{2}\left(  1+\frac{\mathcal{L}^{2}%
}{g_{\varphi\varphi}}\right)   &  =(X_{s}v^{s}+o(v^{s}))^{2}-\left(
\kappa_{p}\left(  1+\frac{L_{H}^{2}}{g_{\varphi H}}\right)  v^{p}%
+o(v^{c})\right)  =\label{ultra_crit_prop}\\
&  =\frac{\kappa_{p}}{A_{p}}(u^{r})_{c}^{2}v^{2c+p-q}+o(v^{2c+p-q})\text{,}%
\end{align}
where $(u^{r})_{c}$ and $c>q/2$ are some constants. It is important to note
that this cancellation can only occur if $s=p/2$. Additionally, this condition
does not involve $\delta$ and is the same as eq. (28) from \cite{ov-zas-prd},
which is a defining property for ultracritical particles when $\delta=0$. The
unusual notation and choice of parameters $(u^{r})_{c}$ and $c$ were made to
simplify the expression for the four-velocity in the case of $\delta=0$. In
this case,%
\begin{equation}
u^{r}=\frac{\sqrt{A}}{N}P\approx(u^{r})_{c}v^{c}.
\end{equation}

However, if $\delta\neq0,$ we get from an equation for $P:$%
\begin{align}
P(v_{e})  &  \approx\sqrt{(\delta+X_{p/2}v_{e}^{p/2})^{2}-\left(  \kappa
_{p}\left(  1+\frac{L_{H}^{2}}{g_{\varphi H}}\right)  v_{e}^{p}\right)
}=\nonumber\\
&  =\sqrt{\delta^{2}+2\delta(X_{p/2}v_{e}^{p/2})+(X_{p/2}v_{e}^{p/2}%
)^{2}-\left(  \kappa_{p}\left(  1+\frac{L_{H}^{2}}{g_{\varphi H}}\right)
v_{e}^{p}\right)  }.
\end{align}

Using (\ref{ultra_crit_prop}) we can write:%
\begin{equation}
P(v_{e})\approx\sqrt{\delta^{2}+2\delta X_{p/2}v^{p/2}+\frac{\kappa_{p}}%
{A_{p}}(u^{r})_{c}^{2}v^{2c+p-q}}=0. \label{p=0_ultracr}%
\end{equation}

Generally, there may be 3 roots: the 1-st is obtained by comparison of the
1-st and 2-nd terms in (\ref{p=0_ultracr}) and is given by
\begin{equation}
v_{t}\approx\left(  -\frac{\delta}{2X_{p/2}}\right)  ^{2/p}.
\label{vt_ultracr}%
\end{equation}

One can easily check that this is the only possible root in (\ref{p=0_ultracr}%
). This root exists in the same cases as the root for a near-subcritical
particle: it exists if $\delta>0$ and $X_{s}<0$, or if $\delta<0$ and
$X_{s}>0$. From (\ref{p=0_ultracr}) we can see that in both these cases a
particle can only move in the range $[0,v_{t}]$. As for all parameters,
(\ref{vt_ultracr}) is closer to the horizon than (\ref{vt_subcr}), we see that
the forward-in-time condition holds only if $\delta>0$ and $X_{s}<0$. (In the
case where $\delta<0$ and $X_{s}>0$, the forward-in-time condition holds for
$v\in\left[  \left(  -\frac{\delta}{X_{s}}\right)  ^{1/s},\infty\right]  $,
while $P$ is real only for $v\in\left[  0,\left(  -\frac{\delta}{2X_{s}%
}\right)  ^{1/s}\right]  $. These regions obviously do not intersect.)

The only remaining cases are $\delta>0$ and $X_{s}>0$, or $\delta<0$ and
$X_{s}<0$. In the first case, both $P$ and $\mathcal{X}$ are positive for all
positions of the particle, thus there are no turning points. In the second
case, $\mathcal{X}$ is negative, which makes this case impossible. Therefore,
we are left with the only case $\delta>0$ and $X_{s}>0$, for which we only
need to find the effective distances at which $P$ changes on the values of the
order $\delta$. Analyzing (\ref{p=0_ultracr}), we see that the dominant
behavior can be obtained by comparing the first and second terms that gives us%
\begin{equation}
v_{e}\sim\delta^{2/p}. \label{vc_ultracr}%
\end{equation}

Summarizing we have:%
\begin{equation}
\text{If particle is NUC (}s<p/2\text{), }v_{c}\in\text{ }\left\{  \text{%
\begin{tabular}
[c]{l}%
$\lbrack0,v_{t}]$ where $v_{t}=\left(  -\frac{\delta}{2X_{s}}\right)  ^{1/s}$
if $\delta>0$ and $X_{s}<0,$\\
$\lbrack0,\infty]$ if $\delta>0$ and $X_{s}>0,$ in this case $v_{e}\sim
\delta^{2/p}.$%
\end{tabular}
}\right.  \label{vt_nuc}%
\end{equation}

\subsection{Near-overcritical particles ($s>p/2$).\label{sec_vt_noc}}

For such particles $v^{p/2}$ term is dominant over $v^{s}$ and (\ref{p=0_eqqq}%
) becomes:%
\begin{equation}
\delta=A_{p/2}v_{t}^{p/2}+o(v_{t}^{p/2}),
\end{equation}
where $A_{p/2}=\sqrt{\kappa_{p}\left(  1+\frac{\mathcal{L}^{2}}{g_{\varphi
\varphi}}\right)  }.$ Solving this equation, we have%
\begin{equation}
v_{t}\approx\left(  \frac{\delta^{2}}{A_{p/2}^{2}}\right)  ^{1/p}%
+o(\delta^{2/p}). \label{vt_overcr}%
\end{equation}

Note that this root of equation $P=0$ exists for all $\delta$ independently of
the forward-in-time condition. In this case, $P$ is given by $\sqrt{\delta
^{2}-A_{p/2}^{2}v^{p}}$ in dominant orders, and solving the equation $P=0$
yields (\ref{vt_overcr}). Additionally, one may observe that $P^{2}$ is
non-negative for $[0,v_{t}]$ regardless of the sign of $\delta$. Now, let us
determine the conditions under which the forward-in-time condition holds in
this range. At first, we note that (\ref{vt_overcr}) is $\sim\delta^{2/p}$,
while the root of the equation $\mathcal{X=}0$, similarly to the case of
near-subcritical particles (\ref{vt_subcr}), is $\sim\delta^{1/s}$. Since for
near-overcritical particles $s>p/2$, (\ref{vt_overcr}) is closer to the
horizon than (\ref{vt_subcr}) (it is of lower order in $\delta$ than
(\ref{vt_subcr})). Thus, in the case of $\delta>0$ and $\mathcal{X}_{s}<0$
(see the first case in (\ref{fit_cond})), the forward-in-time condition holds
throughout the range $\left[  0,\left(  \frac{\delta^{2}}{A_{p/2}^{2}}\right)
^{1/p}\right]  $. If $\ \delta<0$ and $\mathcal{X}_{s}>0$ (see the second case
in (\ref{fit_cond})), the regions of reality of $P$ and positivity of
$\mathcal{X}$ do not intersect, so motion in this case is forbidden. The last
case is $\delta>0$ and $X_{s}>0$ (see the third case in (\ref{fit_cond})). In
this case, the forward-in-time condition holds for all points, and thus the
positivity of $\mathcal{X}$ does not bound regions of particle motion.

In summary, we have:
\begin{equation}
\text{If particle is NOC (}s=p/2\text{), }v_{c}\in\lbrack0,v_{t}]\text{ where
}v_{t}\approx\left(  \frac{\delta^{2}}{A_{p/2}^{2}}\right)  ^{1/p}\text{ if
}\delta>0. \label{vc_overcr}%
\end{equation}

This concludes the analysis of different regions of motion for different types
of particles. It is easy to see that the exact expressions for turning points
$v_{t}$ or effective distances $v_{e}$ are different for each type of
particle. However, for near-critical, near-ultracritical, and
near-overcritical particles, they are of the same order in $\delta$ (see
(\ref{vc_cr}), (\ref{vc_ultracr}), and (\ref{vc_overcr})). Therefore, in this
sense, they may appear to be indistinguishable. To understand the reason why
this classification is still necessary, let us consider the radial component
of the four-velocity in the limit $\delta\rightarrow0$. In this case,
$P\approx\sqrt{(X_{s}v^{s})^{2}-\left(  \kappa_{p}\left(  1+\frac{L_{H}^{2}%
}{g_{\varphi H}}\right)  v^{p}\right)  }$. By using the fact that for
near-subcritical particles $s<p/2$ (and thus the $v^{2s}$ term is dominant),
for near-critical particles $s=p/2$ (and thus the $v^{2s}$ term and the
$v^{p}$ term are of the same order), and for near-ultracritical particles
$s=p/2$ and condition (\ref{ultra_crit_prop}) holds, we obtain for
$|u^{r}|=\frac{\sqrt{A}}{N}P$:%
\begin{align}
\text{Near-subcritical particle}  &  \text{: }|u^{r}|=\sqrt{\frac{A_{q}%
}{\kappa_{p}}}X_{s}v^{\frac{q-p}{2}+s},\label{4-vel-fin-tun_1}\\
\text{Near-critical particle}  &  \text{:}|u^{r}|=\sqrt{\frac{A_{q}}%
{\kappa_{p}}}\sqrt{X_{p/2}^{2}-\kappa_{p}\left(  1+\frac{L_{H}^{2}}{g_{\varphi
H}}\right)  }v^{\frac{q}{2}},\\
\text{Near-ultracritical particle}  &  \text{:}|u^{r}|=(u^{r})_{c}v^{c}.
\label{4-vel-fin-tun_3}%
\end{align}

In all three cases, in the limit $\delta\rightarrow0$ we obtain
correspondingly subcritical, critical, and ultracritical particles, as
introduced in \cite{ov-zas-prd}. As one can see, the behavior of the
four-velocity in these cases is different that justifies the necessity of
distinguishing between near-critical, near-ultracritical, and
near-overcritical particles. In our further analysis, we will also observe
that other physical quantities are different for these particle types. The
only exceptional case is near-overcritical particles that were not considered
in \cite{ov-zas-prd}. This is because when we take the limit $\delta
\rightarrow0$ in this case, $P$ becomes complex near the horizon. Due to this
fact, overcritical particles with $\delta=0$ cannot reach the horizon and
cannot participate in the "pure" BSW phenomenon that was the focus of our
investigation in \cite{ov-zas-prd}. However, non-zero $\delta$ allows such
particles to reach the horizon and thus they are considered in our work.

\section{Different scales of parameters\label{diff_scales}}

Now, we are going to make a next\ step and discuss the interplay of different
parameters in our problem. In the "pure" BSW phenomenon, there is only one
small parameter: the point of collision $v_{c}\ll r_{h}$ (hereafter, the index
"c" denotes the collision point). However, in the case of near-fine-tuned
particles, a new parameter $\delta$ appears, where $\delta\ll1$. To analyze
the different properties of near-fine-tuned particles, we have to specify the
scales of two parameters: $v_{c}$ and $v_{e}$ (or $v_{t}$ if it exists). We
have four different cases (see below) which can be described by different
relations between $v_{c}$ and $v_{t}$ or $v_{e}$ (in each of these cases, the
conditions $\frac{v_{c}}{r_{h}}\ll1$ and $\delta\ll1$ hold).

Before proceeding further, we must make comment on the situation when
$v_{c}<v_{t}$. This means that a particle cannot arrive at the point of
collision from infinity since it would bounce back in the turning point. We
assume that it appears between the horizon and a turning point due to special
initial condition and do not specify their nature. (Say, particle can appear
there due to quantum creation, etc.) In doing so, the interval in which the
scenario develops is very tiny (for nonextremal black holes explicit
expressions for it can be found in eq.(18) of \cite{gr-pav} for the Kerr
metric and in eq. 18) of \cite{prd} for a more general case). Nonetheless, the
BSW effect can indeed exist.

\subsection{1-st case: $v_{c}\gg v_{e}$ (or $v_{t}$)\label{sec_4_a}}

In this case, the point of collision is much further from the horizon than
$v_{e,t}$ (hereafter, the notation $v_{e,t}$ means $v_{e}$ or $v_{t}$ if it
exists). This effectively means that we can take $\delta=0$ while keeping
$\frac{v_{c}}{r_{h}}$ terms in all the quantities we are interested in. To see
why this is so, let us consider the expression for $P$:%
\begin{equation}
P\approx\sqrt{(\delta+X_{s}v_{c}^{s})^{2}-\left(  \kappa_{p}\left(
1+\frac{L_{H}^{2}}{g_{\varphi H}}\right)  v_{c}^{p}\right)  }. \label{p_e>>1}%
\end{equation}

Using that $v_{e,t}\sim\delta^{\max(1/s,2/p)}$ (see (\ref{vc_subcr}),
(\ref{vc_cr}), (\ref{vc_ultracr}) and (\ref{vc_overcr}))$,$ or, inversely,
$\delta\sim v_{e,t}^{\min(s,p/2)},$ we can substitute this to the expression
for (\ref{p_e>>1}) and get:%
\begin{equation}
P\sim\sqrt{(v_{e,t}^{\min(s,p/2)}+X_{s}v_{c}^{s})^{2}-\left(  \kappa
_{p}\left(  1+\frac{L_{H}^{2}}{g_{\varphi H}}\right)  v_{c}^{p}\right)  }.
\end{equation}

We see that the ratio $v_{e}^{\min(s,p/2)}/X_{s}v_{c}^{s}$ tends to zero
because of the condition $v_{c}\gg v_{e}$. This means that $v_{e,t}%
^{\min(s,p/2)}$ term is negligible and dominant terms in (\ref{p_e>>1})
contain $v_{c}$ without $\delta.$ Thus we have:%
\begin{align}
P  &  \approx\sqrt{(X_{s}v_{c}^{s})^{2}-\left(  \kappa_{p}\left(
1+\frac{L_{H}^{2}}{g_{\varphi H}}\right)  v_{c}^{p}\right)  }+O(\delta)=\\
&  \approx\sqrt{\left(  \mathcal{X-}\delta\right)  ^{2}-N^{2}\left(
1+\frac{\mathcal{L}^{2}}{g_{\varphi\varphi}}\right)  }+O(\delta)\text{.}%
\end{align}

One may easily see that the dominant term does not involve $\delta$. Thus, in
this case, the dominant behavior can be obtained by taking the $\delta
\rightarrow0$ limit. This will correspond to a "pure" BSW phenomenon, which
was completely analyzed in \cite{ov-zas-prd}. In this case, the properties of
near-fine-tuned particles are similar to corresponding fine-tuned particles
(for example, near-subcritical particles correspond to subcritical particles,
and so on).

\subsection{2-nd case: $v_{c}\sim v_{e}$ (or $v_{t}$)\label{iv_b}}

In this case, the point of collision is in the same scale of distances as
$v_{e,t}$. Let us obtain expressions for $P(v_{c})$ and $\mathcal{X}(v_{c})$
in the leading order for different types of particles.

\begin{itemize}
\item Near-subcritical particle:

\qquad For such particles, the terms $\delta^{2},$ $v^{2s},$ and $\delta
v^{s}$ are dominant in the expansion for $P^{2}$ (see (\ref{p=0_eqqq})). This
indicates that $\mathcal{X}$ terms prevail over $N^{2}$ in the expansion of
$P$ (see (\ref{P})). Therefore, in this case, we have:
\begin{equation}
P=\sqrt{\mathcal{X}^{2}-N^{2}\left(  1+\frac{\mathcal{L}^{2}}{g_{\varphi
\varphi}}\right)  }\approx\mathcal{X-}\frac{N^{2}}{2\mathcal{X}}\left(
1+\frac{\mathcal{L}^{2}}{g_{\varphi\varphi}}\right)  . \label{p_nsc_cond}%
\end{equation}

\qquad In this expression, we have also included higher-order terms because,
as we will demonstrate it below in the analysis of the energy of the
collision, the dominant terms will cancel each other out, and the energy will
be determined by higher-order corrections.

\qquad If there is no turning point (that occurs only when $\delta>0$ and
$X_{s}>0$, see the 3rd condition in (\ref{fit_cond})), the dominant term in
$P$ (specifically $\mathcal{X}$) at $v_{c}\sim v_{e}$ takes the following
form:%
\begin{equation}
P\approx\mathcal{X}\approx\delta+X_{s}v_{c}^{s}+o(\delta). \label{px_nsc}%
\end{equation}

\qquad Using the approximation $v_{c}\sim v_{e}\sim\delta^{1/s}$, we observe
that $P\sim\mathcal{X\sim}$ $\delta$. However, if a turning point exists, the
situation becomes more complex. In this case, we can invert (\ref{vt_subcr})
and obtain $\delta\approx-X_{s}v_{t}^{s}$. By substituting this into the
expression for $\mathcal{X}$, we have:
\begin{equation}
\mathcal{X}\approx\delta+X_{s}v_{c}^{s}\approx X_{s}(v_{c}^{s}-v_{t}^{s}).
\end{equation}

\qquad One can easily see that since $v_{c}\sim v_{e}$, $P$ is of the order
$\delta$. However, if $v_{c}$ approaches $v_{t}$ (this means that the
collision point reaches the turning point), the quantity $\mathcal{X}$ tends
to zero. To describe the small difference that arises in this case, we make
the assumption:
\begin{equation}
\delta=-X_{s}v_{c}^{s}+B_{r}v_{c}^{r}, \label{delt_spec_cond}%
\end{equation}

(in fact we could assume that difference $v_{c}-v_{t}$ is some small parameter
but further analysis will be simpler if we consider more concrete example
according to (\ref{delt_spec_cond})). In this case
\begin{equation}
\mathcal{X}\approx B_{r}v_{c}^{r}, \label{x_b_nsc}%
\end{equation}%
\begin{equation}
P\approx B_{r}v_{c}^{r}-\frac{A_{p/2}^{2}}{2B_{r}}v_{c}^{p-r}. \label{p_b_nsc}%
\end{equation}

\qquad From the expression for $\mathcal{X}$, we can observe that the
forward-in-time condition holds only if $B_{r}>0$. On the other hand, from the
expression for $P$, we can see that the expansion (\ref{p_nsc_cond}) holds
only if $s<r<p/2$. It may seem strange that this special case needs to be
considered. However, as we will demonstrate it in the analysis of the energy
of the collision, the behavior of energy becomes quite special when $v_{c}$
approaches $v_{t}$.

\qquad In further analysis, we will need to know the behavior of the quantity
$\sqrt{\mathcal{X}^{2}-N^{2}}$. Considering that in this case $s<p/2$ and thus
the $\mathcal{X}$ term is dominant, we have:%
\begin{equation}
\sqrt{\mathcal{X}^{2}-N^{2}}\approx\mathcal{X-}\frac{N^{2}}{2\mathcal{X}}.
\label{sqrt_nsc}%
\end{equation}

\qquad As one can see, the behaviour of this quantity is similar to the
beaviour of $P.$

\item Near-critical particle:

\qquad For these particles situation is more complicated because all terms are
comparable, and we have (in the main order of $\delta$):%
\begin{equation}
P\approx\sqrt{(\delta+X_{p/2}v_{c}^{p/2})^{2}-A_{p/2}^{2}v_{c}^{p}}.
\label{p_nc_cond}%
\end{equation}

\qquad If there is no turning point, we take into account that $v_{c}\sim
v_{e}\sim\delta^{2/p}$ (see (\ref{vc_cr})) and we see that $P\sim\delta.$

\qquad In this case we also have
\begin{equation}
\mathcal{X}\approx\delta+X_{p/2}v_{c}^{p/2}\sim\delta. \label{x_nc_cond}%
\end{equation}

\qquad For $\sqrt{\mathcal{X}^{2}-N^{2}}$, we can observe that both
$\mathcal{X}^{2}$ and $N^{2}$ are of the same order, and we have
$\sqrt{\mathcal{X}^{2}-N^{2}}\sim\delta$.

\qquad If there are roots of equation $P=0$ (that occurs if $\delta>0$ and
$A_{p/2}>X_{p/2}$ or $\delta<0$ and $A_{p/2}<X_{p/2}$, see (\ref{vt_nc})), the
expression for $P$ becomes (by inverting (\ref{vt_cr}) that gives us
$\delta=(A_{p/2}-X_{p/2})v_{t}^{p/2}$):%
\begin{align}
P  &  \approx\sqrt{\delta^{2}+2\delta X_{p/2}v_{c}^{p/2}+(X_{p/2}^{2}%
-A_{p/2}^{2})v_{c}^{p}}\approx\\
&  \sqrt{(A_{p/2}-X_{p/2})^{2}v_{t}^{p}+2(A_{p/2}-X_{p/2})X_{p/2}v_{t}%
^{p/2}v_{c}^{p/2}+(X_{p/2}^{2}-A_{p/2}^{2})v_{c}^{p}}.
\end{align}

\qquad Note that if $v_{c}\sim v_{t}\sim\delta^{2/p},$ we have $P\sim\delta.$
However, if $v_{c}\rightarrow v_{t}$ one can check that $P\rightarrow0.$ In
this case we assume
\begin{equation}
\delta=(A_{p/2}-X_{p/2})v_{c}^{p/2}+B_{r}v_{c}^{r}, \label{delta_spec_cond_cr}%
\end{equation}

where $r>p/2.$

\qquad Substituting this to $P$ we have:%
\begin{align}
P  &  \approx\sqrt{(\delta+X_{p/2}v_{c}^{p/2})^{2}-A_{p/2}^{2}v_{c}^{p}}%
=\sqrt{(A_{p/2}v_{c}^{p/2}+B_{r}v_{c}^{r})^{2}-A_{p/2}^{2}v_{c}^{p}}\approx\\
&  \approx\sqrt{2A_{p/2}B_{r}v_{c}^{p/2+r}+B_{r}^{2}v_{c}^{2r}}\approx
\sqrt{2A_{p/2}B_{r}}v_{c}^{p/4+r/2}. \label{p_b_cr}%
\end{align}

\qquad This will not have any special consequences in the behaviour of energy
of collision but will influence the behaviour of acceleration.

\item Near-ultracritical particle:

\qquad In this case we can use (\ref{ultra_crit_prop}) and (\ref{vc_ultracr})
that gives us:%
\begin{equation}
P\approx\sqrt{\delta^{2}+2\delta X_{p/2}v_{c}^{p/2}}. \label{p_nuc_cond}%
\end{equation}
\qquad\qquad\qquad\qquad\ 

\qquad If there is no turning point, using (\ref{vc_ultracr}) we have
$v_{c}\sim v_{e}\sim\delta^{2/p}.$ Substituting this to (\ref{p_nuc_cond}) we
see that $P\sim\delta.$

\qquad In this case we also have%
\begin{equation}
\mathcal{X}\approx\delta+X_{p/2}v_{c}^{p/2}\sim\delta. \label{x_nuc_cond}%
\end{equation}

\qquad The same holds for $\sqrt{\mathcal{X}^{2}-N^{2}}:$%
\begin{equation}
\sqrt{\mathcal{X}^{2}-N^{2}}\approx\sqrt{(\delta+X_{p/2}v_{c}^{p/2}%
)^{2}-\left(  \kappa_{p}v_{c}^{p}+o(v_{c}^{p})\right)  }\sim\delta.
\end{equation}

\qquad If a turning point exists, inverting (\ref{vt_nuc}) we have
$\delta=-2X_{s}v_{t}^{s}.$ Substituting this in the expression for $P$ we
have:%
\begin{equation}
P\approx2X_{p/2}\sqrt{v_{t}^{p/2}(v_{t}^{p/2}-v_{c}^{p/2})}.
\end{equation}

\qquad If $v_{c}\sim v_{t}$ we see that $P\sim\delta.$ While if $v_{c}%
\rightarrow v_{t},$ $P\rightarrow0.$ In this case let us write
\begin{equation}
\delta=-2X_{p/2}v_{c}^{p/2}+B_{r}v_{c}^{r}, \label{delta_spec_cond_ucr}%
\end{equation}

where $r>p/2.$

\qquad Substituting this in the expression for $P$ we have:%
\begin{align}
P  &  \approx\sqrt{(-2X_{p/2}v_{c}^{p/2}+B_{r}v_{c}^{r})^{2}+2(-2X_{p/2}%
v_{c}^{p/2}+B_{r}v_{c}^{r})X_{p/2}v_{c}^{p/2}}\approx\label{p_b_ultra}\\
&  \approx\sqrt{2X_{p/2}B_{r}v_{c}^{p/2+r}+B_{r}^{2}v_{c}^{2r}}\approx
\sqrt{2X_{p/2}B_{r}}v_{c}^{p/4+r/2}.
\end{align}

\qquad However, we will see that in this case the energy of the collision does
not change drastically depending on $v_{c}$.

\item Near-overcritical particle:

\qquad In this case only $\delta^{2}$ and $v_{c}^{p}$ terms are dominant (see
(\ref{vc_overcr})) and we get:%
\begin{equation}
P\approx\sqrt{\delta^{2}-A_{p/2}^{2}v_{c}^{p}}. \label{p_noc_cond}%
\end{equation}

\qquad Now, the turning point always exists. Inverting (\ref{vt_overcr}) we
have $\delta=A_{p/2}v_{t}^{p/2}.$ Substituting this in $P$ we obtain:%
\begin{equation}
P\approx\sqrt{A_{p/2}^{2}(v_{t}^{p}-v_{c}^{p})}.
\end{equation}

\qquad If $v_{c}\sim v_{t}$ we see that $P\sim\delta.$ While if $v_{c}%
\rightarrow v_{t},$ $P\rightarrow0.$ In this case let us write%
\begin{equation}
\delta=A_{p/2}v_{c}^{p/2}+K_{r}v_{c}^{r}, \label{delta_spec_cond_ocr}%
\end{equation}

where $r>p/2.$ Substituting this in the expression for $P$ we have:%
\begin{align}
P  &  \approx\sqrt{\delta^{2}-A_{p/2}^{2}v_{c}^{p}}\approx\sqrt{(A_{p/2}%
v_{c}^{p/2}+K_{r}v_{c}^{r})^{2}-A_{p/2}^{2}v_{c}^{p}}\approx\label{p_b_over}\\
&  \approx\sqrt{2A_{p/2}B_{r}v_{c}^{p/2+r}+B_{r}^{2}v_{c}^{2r}}\approx
\sqrt{2A_{p/2}B_{r}}v_{c}^{p/4+r/2}.
\end{align}

\qquad However, we will see that in this case the energy of the collision does
not change drastically depending on $v_{c}$.
\end{itemize}

\subsection{3-rd case: $v_{c}\ll v_{e}$\label{iv_c}}

Now, the point of collision is much closer to the horizon than $v_{e}$. This
case can be obtained simply by taking the limit $v_{c}\rightarrow0$ while
keeping terms with $\delta$. To see this, we have to use the fact that
$v_{e,t}\sim\delta^{\max(1/s,2/p)}$ (see (\ref{vc_subcr}), (\ref{vc_cr}),
(\ref{vc_ultracr}) and (\ref{vc_overcr})), or, inversely, $\delta\sim
v_{e,t}^{\min(s,p/2)}$. We can substitute this into the expression for
(\ref{p_e>>1}) and get:%
\begin{equation}
P\sim\sqrt{(v_{e,t}^{\min(s,p/2)}+X_{s}v_{c}^{s})^{2}-\left(  \kappa
_{p}\left(  1+\frac{L_{H}^{2}}{g_{\varphi H}}\right)  v_{c}^{p}\right)  .}%
\end{equation}
\qquad\qquad\ \ \qquad

First of all, we note that the ratio of the 2nd and 1st terms in $P$ is
$\frac{v_{c}^{s}}{v_{e,t}^{\min(s,p/2)}}$. Since we assume $v_{c}\ll v_{e}$,
the second term is much less than the first one. Additionally, we observe that
the 3rd term is much less than the 1st term because $\frac{v_{c}^{p}}%
{v_{e,t}^{2\min(s,p/2)}}\ll1$ due to $v_{c}\ll v_{e}$. Therefore, the first
term is dominant and we can write:
\begin{equation}
P=\sqrt{\mathcal{X}^{2}-N^{2}\left(  1+\frac{\mathcal{L}^{2}}{g_{\varphi
\varphi}}\right)  }\approx\mathcal{X-}\frac{N^{2}}{2\mathcal{X}}\left(
1+\frac{\mathcal{L}^{2}}{g_{\varphi\varphi}}\right)  . \label{p_vc<<vt}%
\end{equation}

We keep here higher order corrections because they will be important for
analysis of energy of collision.

\bigskip Also we note that in this case
\begin{equation}
\sqrt{\mathcal{X}^{2}-N^{2}}\approx\mathcal{X-}\frac{N^{2}}{2\mathcal{X}}.
\label{sqrt_vc<<vt}%
\end{equation}

Generally speaking, this corresponds to usual particles because all entries of
the point of collision in the expression for $P$ are much less than the value
of $\mathcal{X}$ on the horizon.

The main idea of this section is that there are 4 cases of possible interplay
between the small parameters $v_{c}$ and $\delta$ (actually, we saw that the
classification is mainly defined by the relations between $v_{c}$ and $\delta$
raised to some power).

The first case occurs when $v_{e,t}\ll v_{c}$ (but $v_{c}\ll r_{h}$ still
holds) that, as we showed, corresponds to the "pure" BSW phenomenon, which is
not of interest in this work.

The second case is $v_{e,t}\sim v_{c}$, for which, as we showed, $P,$
$\mathcal{X},$ and $\sqrt{\mathcal{X}^{2}-N^{2}}$ are $\sim\delta$ (exact
expressions can be found in the corresponding parts of the text).

In the case $v_{c}\ll v_{e,t}$, all $v_{c}$ terms in the expressions for $P,$
$\mathcal{X}$ and $\sqrt{\mathcal{X}^{2}-N^{2}}$ are negligible that
corresponds to the case of usual particles that have already been investigated.

\subsection{4-th case: $v_{c}\rightarrow v_{t}$}

The last case is possible if the particles impose the existence of turning
points and when $v_{c}\rightarrow v_{t}$. In this case, depending on the type
of particle, we assume that $\delta$ is given by (\ref{delt_spec_cond}),
(\ref{delta_spec_cond_cr}), (\ref{delta_spec_cond_ucr}) or
(\ref{delta_spec_cond_ocr}). As we showed, in these cases $P$ is either $\sim
v_{c}^{r}$ (for near-subcritical particles) or $\sim v_{c}^{p/4+r/2}$ (for
near-critical, near-ultracritical, or near-overcritical particles).

To summarize, new scenarios of particle's motion related to the non-zero
$\delta$ can only be obtained if $v_{e,t}\sim v_{c}$ or $v_{c}\rightarrow
v_{t}$. All other ranges of the $v$ coordinate correspond to already
investigated cases \cite{ov-zas-prd}.

\section{Energy of collision\label{energy}}

\subsection{General relations}

As we mentioned above, we are mainly interested in the possibility of the BSW
phenomenon that is related to an unbounded growth of energy in the center of
mass frame of two colliding particles. This energy is given by%
\begin{equation}
E_{c.m.}^{2}=-(m_{1}u_{1\mu}+m_{2}u_{2\mu})(m_{1}u_{1}^{\mu}+m_{2}u_{2}^{\mu
})=m_{1}^{2}+m_{2}^{2}-2m_{1}m_{2}u_{1}^{\mu}u_{2\mu}\text{,}%
\end{equation}
where $\gamma=-u_{1\mu}u^{2\mu}$ is the Lorentz gamma factor of relative
motion. Substituting the expression for the four-velocity (\ref{4_vel}), we
have%
\begin{equation}
\gamma=\frac{\mathcal{X}_{1}\mathcal{X}_{2}-P_{1}P_{2}}{N^{2}}-\frac
{\mathcal{L}_{1}\mathcal{L}_{2}}{g_{\varphi\varphi}}. \label{gamma}%
\end{equation}

Hereafter, we assume that both particles move toward the horizon, so
$\sigma_{1}=\sigma_{2}=-1$.

The second term in (\ref{gamma}) is regular, so we are interested, when the
first one is unbounded.

Let us discuss all possible cases of particle collision depending on types of
particles. Cases of collision between fine-tuned particles with usual or other
fine-tuned particles have already been discussed in \cite{ov-zas-prd}. Thus,
we are left with a discussion of the collision between near-fine-tuned
particles with fine-tuned or usual particles, as well as the cases where both
particles participating in the collision are near-fine-tuned.

\subsection{1-st particle is fine-tuned (or usual), 2-nd is near fine-tuned.}

Let us start with the analysis of the case in which one particle (let us call
this particle 1) is fine-tuned (or usual) and particle 2 is near fine-tuned.
Before we proceed further, let us remind a reader several properties of
fine-tuned particles. Different types of particles are defined through their
expansion of $\mathcal{X}$. Generally, $\mathcal{X}$ for fine-tuned particles
has an expansion in the form%
\begin{equation}
\mathcal{X=}X_{s}v^{s}+o(v^{s}), \label{X_us}%
\end{equation}
where for usual particles $s=0$, for subcritical $0<s<p/2$, for critical
$s=p/2$, for ultracritical $s=p/2$, and the condition (\ref{ultra_crit_prop})
has to hold. In further analysis, we use the abbreviations "U" for usual
particle, "SC" for subcritical, "C" for critical, and "UC" for ultracritical.

For usual and subcritical particles, as $\frac{N^{2}}{\mathcal{X}}%
\rightarrow0$ as $v\rightarrow0$, we can expand the function $P$ and obtain:
\begin{equation}
P=\mathcal{X}-\frac{N^{2}}{2\mathcal{X}}\left(  \frac{\mathcal{L}^{2}%
}{g_{\varphi\varphi}}+1\right)  +...=\mathcal{X}+O(v^{p-s}). \label{P_subcr}%
\end{equation}

For critical particles $N^{2}$ and $\mathcal{X}$ are of the same order, so we
have%
\begin{equation}
P=P_{p/2}v^{p/2}+..., \label{P_cr}%
\end{equation}

where $P_{p/2}=\sqrt{X_{p/2}^{2}-\kappa_{p}\left(  \frac{L_{H}^{2}}{g_{\varphi
H}}+1\right)  }.$

For ultracritical particles
\begin{equation}
P=P_{c+(p-q)/2}v^{^{c+\frac{p-q}{2}}}+..., \label{P_ultracr}%
\end{equation}
where $c>q/2$ and $P_{c+(p-q)/2}=\sqrt{\frac{\kappa_{p}}{A_{q}}}(u^{r})_{c}$
(note that these expansions correspond to the ones obtained by taking the
limit $\delta\rightarrow0$ for near-fine-tuned particles, see
(\ref{4-vel-fin-tun_1}-\ref{4-vel-fin-tun_3})).

Now, let us discuss the properties of particle 2. As concluded at the end of
Section \ref{diff_scales}, the only interesting cases are those when for the
second particle either $v_{c}\sim v_{e,t}$ or $v_{c}\rightarrow v_{t},$ so we
only need to consider these cases.

Next, consider the collision of two particles. If the first particle is usual
or subcritical, we substitute (\ref{P_subcr}) and (\ref{X_us}) to
(\ref{gamma}) and get:%
\begin{equation}
\gamma\approx\frac{X_{s_{1}}^{(1)}v_{c}^{s_{1}}[\mathcal{X}_{2}-P_{2}]}{N^{2}%
}+\frac{P_{2}}{2X_{s_{1}}^{(1)}v_{c}^{s_{1}}}\left(  \frac{L_{H1}^{2}%
}{g_{\varphi\varphi}}+1\right)  \approx\frac{X_{s}^{(1)}}{\kappa_{p}}%
\frac{[\mathcal{X}_{2}-P_{2}]}{v_{c}^{p-s_{1}}}+\frac{\left(  \frac{L_{H1}%
^{2}}{g_{\varphi\varphi}}+1\right)  }{2X_{s_{1}}^{(1)}}\frac{P_{2}}%
{v_{c}^{s_{1}}}, \label{gamma_subcr}%
\end{equation}
where upper index $(1)$ means quantities related to the 1-st particle
(fine-tuned). We will postpone analysis of this complicated expression to the
next subsections.

If the 1-st particle is critical, then we substitute (\ref{P_cr}) and
(\ref{X_us}) to (\ref{gamma}) and have:%

\begin{equation}
\gamma\approx\frac{1}{\kappa_{p}}\frac{X_{p/2}^{(1)}\mathcal{X}_{2}%
-P_{p/2}^{(1)}P_{2}}{v_{c}^{p/2}}. \label{gamma_cr}%
\end{equation}

If the 1-st particle is ultracritical,%

\begin{equation}
\gamma\approx\frac{1}{\kappa_{p}}\frac{X_{s_{1}}^{(1)}\mathcal{X}_{2}}%
{v_{c}^{p/2}} \label{gamma_ultracr}%
\end{equation}

(note that in this case term with $P_{1}$ is absent because $P_{1}$ is of
higher order in $v_{c}$ then $\mathcal{X}_{1}$).

\subsubsection{Near-subcritical particles}

Now, let us analyze the behavior of the gamma factor concerning various types
of the second particle. We begin by considering a scenario where particle 2 is
near-subcritical. The initial case for analysis involves the situation where
particle 1 is usual or subcritical. Before delving into the analysis of the
behavior of $\gamma$, it is important to note that in this case (using
(\ref{p_nsc_cond}))
\begin{equation}
\mathcal{X}_{2}-P_{2}=\frac{N^{2}}{2\mathcal{X}_{2}}\left(  1+\frac
{\mathcal{L}_{2}^{2}}{g_{\varphi\varphi}}\right)  .
\end{equation}

Substituting this in (\ref{gamma_subcr}) we have:%
\begin{equation}
\gamma\approx\frac{X_{s_{1}}^{(1)}v_{c}^{s_{1}}}{2\mathcal{X}_{2}}\left(
1+\frac{L_{H2}^{2}}{g_{\varphi\varphi}}\right)  +\frac{\left(  \frac
{L_{H1}^{2}}{g_{\varphi\varphi}}+1\right)  }{2X_{s_{1}}^{(1)}}\frac
{\mathcal{X}_{2}}{v_{c}^{s_{1}}}.
\end{equation}

Now, let us analyze the various possible cases. As demonstrated in Section
\ref{diff_scales}, the behavior of $\mathcal{X}_{2}$ depends on whether
$v_{c}\sim v_{e,t}$ or $v_{c}\rightarrow v_{t}.$ If $v_{c}\sim v_{e,t},$ we
can use (\ref{vc_subcr}) (or (\ref{fit_cond}) if a turning point exists) and
get $v_{c}\sim\delta^{1/s_{2}}$ (or, inverting, $\delta\sim v_{c}^{s_{2}}$).
Furthermore, using the relation $\mathcal{X}_{2}\sim\delta$ when $v_{c}\sim
v_{e,t}$, we obtain two terms in the expression for $\gamma$: the first one is
proportional to $\sim v_{c}^{s_{1}-s_{2}},$ while the second term is
proportional to $\sim v_{c}^{s_{2}-s_{1}}.$ The dominant term is determined by
the smaller degree in $v_{c}$ between these two. Combining these terms, we can
find that $\gamma\sim v_{c}^{-|s_{2}-s_{1}|}.$

If the turning point exists (this, as one can see from (\ref{fit_cond}),
happens if $\delta^{(2)}>0$ and $X_{s}^{(2)}<0$ or if $\delta^{(2)}<0$ and
$X_{s}^{(2)}>0$) and if $v_{c}\rightarrow v_{t}$, it follows from
(\ref{x_b_nsc}) that $\mathcal{X}_{2}\approx B_{r_{2}}^{(2)}v_{c}^{r_{2}}.$
Substituting this in (\ref{gamma_subcr}) we get:%
\begin{equation}
\gamma\approx\frac{X_{s_{1}}^{(1)}v_{c}^{s_{1}}}{2B_{r_{2}}^{(2)}v_{c}^{r_{2}%
}}\left(  1+\frac{L_{H2}^{2}}{g_{\varphi\varphi}}\right)  +\frac{\left(
\frac{L_{H1}^{2}}{g_{\varphi\varphi}}+1\right)  }{2X_{s_{1}}^{(1)}}%
\frac{B_{r_{2}}^{(2)}v_{c}^{r_{2}}}{v_{c}^{s_{1}}}.
\end{equation}

As one can see, there are two terms: the first is $\sim v_{c}^{s_{1}-r_{2}}$,
while the second is $\sim v_{c}^{r_{2}-s_{1}}$. Since the divergence of the
gamma factor is defined by the dominant term, we can combine these two terms
and write $\gamma\sim v_{c}^{-|s_{1}-r_{2}|}$.

All these expressions can be formulated briefly under this condition:%
\begin{equation}
\text{If 1-st particle is U or SC and 2-nd-NSC, then }\left\{  \text{%
\begin{tabular}
[c]{l}%
$\gamma\sim v_{c}^{-|s_{2}-s_{1}|}\text{ if }v_{c}\sim v_{e,t},$\\
$\gamma\sim v_{c}^{-|s_{1}-r_{2}|}$ if $v_{c}\rightarrow v_{t}$ (see
(\ref{delt_spec_cond})).
\end{tabular}
}\right.  \text{ } \label{gamma_1us_2_nsc}%
\end{equation}

In the case if 1-st particle is critical (see (\ref{gamma_cr})) or
ultracritical (see (\ref{gamma_ultracr})) and if $v_{c}\sim v_{e,t}$ we can
use that $\mathcal{X}_{2}\sim\delta$ (see (\ref{px_nsc}) and discussion after
it) and get $\gamma\sim\frac{\delta}{v_{c}^{p/2}}.$ Now, $v_{c}\sim
\delta^{1/s_{2}}$ (see (\ref{vc_subcr}) and (\ref{vt_subcr})). Inverting it,
we have $\delta\sim v_{c}^{s_{2}},$ so we can write $\gamma\sim v_{c}%
^{s_{2}-p/2}.~$

If the turning point exists (this, as one can see from (\ref{fit_cond}),
happens if $\delta^{(2)}>0$ and $X_{s}^{(2)}<0$ or if $\delta^{(2)}<0$ and
$X_{s}^{(2)}>0$) and if $v_{c}\rightarrow v_{e},$ we can use (\ref{x_b_nsc})
and write $\mathcal{X}_{2}\approx B_{r_{2}}^{(2)}v_{c}^{r_{2}}$. If particle 1
is critical, we can use (\ref{gamma_cr}) and write
\begin{equation}
\gamma\approx\frac{1}{\kappa_{p}}\frac{(X_{p/2}^{(1)}-P_{p/2}^{(1)})B_{r_{2}%
}^{(2)}v_{c}^{r_{2}}}{v_{c}^{p/2}}.
\end{equation}

From this expression we see easily that $\gamma\sim v_{c}^{r_{2}-p/2}.$ The
same holds for case when 1-st particle is ultracritical (to see this one has
to substitute (\ref{x_b_nsc}) to (\ref{gamma_ultracr})). To summarize, we
have:%
\begin{equation}
\text{If 1-st particle is C or UC and 2-nd-NSC., then}\left\{  \text{
\begin{tabular}
[c]{l}%
$\gamma\sim v_{c}^{s_{2}-\frac{p}{2}}\text{ if }v_{c}\sim v_{e},$\\
$\gamma\sim v_{c}^{r_{2}-\frac{p}{2}}$ if $v_{c}\rightarrow v_{t}$ (see
(\ref{delt_spec_cond})).
\end{tabular}
}\right.  \text{ } \label{gamma_1_cr_2nsc}%
\end{equation}

\subsubsection{Near-critical, near-ultracritical and near-overcritical
particles}

Now let us consider the cases when the 2-nd particle is near-critical,
near-ultracritical or near-overcritical. Before we proceed further, let us
consider $\mathcal{X}_{2}-P_{2}$ for a different types of particles. We start
with a case of near-critical particles. Using (\ref{p_nc_cond}) and
(\ref{x_nc_cond}) we have:%

\begin{equation}
\mathcal{X}_{2}-P_{2}\approx(\delta+X_{p/2}v_{c}^{p/2})-\sqrt{(\delta
+X_{p/2}v_{c}^{p/2})^{2}-\kappa_{p}\left(  1+\frac{L_{H}^{2}}{g_{\varphi H}%
}\right)  v_{c}^{p}}.
\end{equation}

If $v_{c}\sim v_{e,t},$ corresponding terms do not cancel each other and,
using (\ref{vc_cr}), we have $\mathcal{X}_{2}-P_{2}\sim\delta$ (note that this
also holds for the case $v_{c}\rightarrow v_{t}$ because in this limit
$P_{2}\rightarrow0$ while $\mathcal{X}_{2}$ remains $\sim\delta$)$.$

In the case of near-overcritical particles situation is somehow similar. Using
(\ref{p_noc_cond}) we have:%
\begin{equation}
\mathcal{X}_{2}-P_{2}\approx\delta-\sqrt{\delta^{2}-\kappa_{p}\left(
1+\frac{L_{H}^{2}}{g_{\varphi H}}\right)  v_{c}^{p}}.
\end{equation}

If $v_{c}\sim v_{e,t}$ (and $v_{c}\rightarrow v_{t}$) corresponding terms do
not cancel each other and, using (\ref{vc_overcr}), we have $\mathcal{X}%
_{2}-P_{2}\sim\delta.$

In the case of near-ultracritical particles the situation is also similar.
Using (\ref{p_nuc_cond}) and (\ref{x_nuc_cond}), we get:%
\begin{equation}
\mathcal{X}_{2}-P_{2}\approx\left(  \delta+X_{p/2}v_{c}^{p/2}\right)
-\sqrt{\delta^{2}+2\delta X_{p/2}v_{c}^{p/2}}.
\end{equation}

If $v_{c}\sim v_{e,t}$ (and $v_{c}\rightarrow v_{t}$), the corresponding terms
do not cancel each other, and using (\ref{vc_ultracr}), we have $\mathcal{X}%
_{2}-P_{2}\sim\delta$.

All these cases are similar in the sense that if $v_{c}\sim v_{e,t}$ (or
$v_{c}\rightarrow v_{t}$), then $\mathcal{X}_{2}-P_{2}\sim\delta$. Using these
facts, we are ready to analyze the gamma-factor. We start with a case when the
first particle is usual or subcritical (see (\ref{gamma_subcr})). If
$v_{c}\sim v_{e,t}$, then we have two terms: the first one is $\sim
\frac{\delta}{v_{c}^{p-s_{1}}}$, while the second one is $\sim\frac{\delta
}{v_{c}^{s_{1}}}$. Since for usual or subcritical particles, $0<s_{1}<p/2$,
the first term is dominant. Using the fact that $v_{c}\sim v_{e,t}$ and
equations (\ref{vc_cr}), (\ref{vc_overcr}), or (\ref{vc_ultracr}), we have
$v_{c}\sim\delta^{2/p}$, or inversely, $\delta\sim v_{c}^{p/2}$. Substituting
this into (\ref{gamma_subcr}), we get $\gamma\sim v_{c}^{-(\frac{p}{2}-s_{1}%
)}$. To summarize, we have:%
\begin{equation}
\text{If 1-st particle is U or SC and 2-nd-NC, NUC or NOC, then }\gamma\sim
v_{c}^{-(\frac{p}{2}-s_{1})}\text{ if }v_{c}\sim v_{e,t}\text{ .}
\label{gamma_1_us_2_nc}%
\end{equation}

If particle 1 is critical (see (\ref{gamma_cr})) of ultracritical (see
(\ref{gamma_ultracr})), $\gamma\sim\frac{\delta}{v_{c}^{p/2}}~$if $v_{c}\sim
v_{e,t}$. Reverting (\ref{vc_cr}) or (\ref{vc_overcr}) we can write that in a
case $v_{c}\sim v_{e,t}$, $\delta\sim v_{c}^{p/2}.$ Substituting this in the
expression for gamma factor we have:%
\begin{equation}
\text{If 1-st particle is C or UC and 2-nd-NC, NUC or NOC, then }%
\gamma=O(1)\text{ if }v_{c}\sim v_{e,t}\text{ or if }v_{c}\rightarrow v_{t}.
\label{gamma_1_cr_2_nc}%
\end{equation}

\subsection{1-st and 2-nd particles are near fine-tuned.}

First of all, let us formulate which cases we have to consider. As we
concluded at the end of Section \ref{diff_scales}, the only new cases are such
that for both particles either $v_{c}\sim v_{e,t}$ or $v_{c}\rightarrow
v_{t}.$ Let us consider different subcases:

\subsubsection{1-st and 2-nd particles are near-subcritical}

In this case we can use (\ref{p_nsc_cond}) for both particles and get%
\begin{equation}
\gamma\approx\frac{\mathcal{X}_{1}}{2\mathcal{X}_{2}}\left(  1+\frac
{\mathcal{L}_{2}^{2}}{g_{\varphi\varphi}}\right)  +\frac{\mathcal{X}_{2}%
}{2\mathcal{X}_{1}}\left(  1+\frac{\mathcal{L}_{1}^{2}}{g_{\varphi\varphi}%
}\right)  .
\end{equation}

Additionally, there are 4 subcases.

\begin{itemize}
\item If both particles satisfy $v_{c}^{(1,2)}\sim v_{e,t}^{(1,2)}$,
(\ref{px_nsc}) holds for both particles, and $\gamma$ is given by two terms:
the first term is $\sim\frac{\delta_{1}}{\delta_{2}}$ and the second one is
$\sim\frac{\delta_{2}}{\delta_{1}}$. Using the fact that $\delta_{1,2}\sim
v_{c}^{s_{1,2}}$ for both particles, we can combine both terms and write
$\gamma\sim v_{c}^{-|s_{1}-s_{2}|}$.

\item If for the first particle $v_{c}^{(1)}\sim v_{e,t}^{(1)}$, while for the
second particle $v_{c}^{(2)}\rightarrow v_{t}^{(2)}$, we can use
(\ref{px_nsc}) for the first particle and (\ref{x_b_nsc}) for the second
particle. There are two terms: the first one is $\sim\frac{\delta_{1}}%
{v_{c}^{r_{2}}}$, while the second one is $\sim\frac{v_{c}^{r_{2}}}{\delta
_{1}}$. Since for the first particle $\delta_{1}\sim v_{c}^{s_{1}}$, we can
combine these two terms and get $\gamma\sim v_{c}^{-|s_{1}-r_{2}|}$.

\item If for the first particle $v_{c}^{(1)}\rightarrow v_{t}^{(1)}$, while
for the second particle $v_{c}^{(2)}\sim v_{e,t}^{(2)}$, we can use the
expressions from the previous subcase and write $\gamma\sim v_{c}%
^{-|r_{1}-s_{2}|}$.

\item If for both particles $v_{c}^{(1,2)}\rightarrow v_{t}^{(1,2)}$, we can
use (\ref{x_b_nsc}) for them. We have two terms: the first one is $\sim
\frac{v_{c}^{r_{1}}}{v_{c}^{r_{2}}}$, while the second one is $\sim\frac
{v_{c}^{r_{2}}}{v_{c}^{r_{1}}}$. We can combine them and write $\gamma\sim
v_{c}^{-|r_{1}-r_{2}|}$.
\end{itemize}

\subsubsection{1-st particle is NSC while 2-nd is NC, NUC or NOC}

For the first particle we can use (\ref{p_nsc_cond}) and write
\begin{equation}
\gamma\approx\frac{\mathcal{X}_{1}}{\kappa_{p}}\frac{[\mathcal{X}_{2}-P_{2}%
]}{v_{c}^{p}}+\frac{\left(  \frac{L_{H1}^{2}}{g_{\varphi\varphi}}+1\right)
}{2\mathcal{X}_{1}}P_{2}.
\end{equation}

Here, there are several cases.

\begin{itemize}
\item For both particles $v_{c}^{(1,2)}\sim v_{e,t}^{(1,2)}$. In this case, we
can use (\ref{px_nsc}) for the first particle and (\ref{p_nc_cond}),
(\ref{p_nuc_cond}), or (\ref{p_noc_cond}) for the second particle. (All of
these cases, in fact, give the same result because $P_{2}\sim\mathcal{X}%
_{2}\sim\delta_{2}$). Substituting this into the expressions for $\gamma$, we
obtain two terms: the first term is of the order $\sim\frac{\delta_{1}%
\delta_{2}}{v_{c}^{p}}$, while the second one is $\sim\frac{\delta_{2}}%
{\delta_{1}}$. Using the fact that $\delta_{1}\sim v_{c}^{s_{1}}$ for the
first particle and $\delta_{2}\sim v_{c}^{p/2}$ for the second particle, we
see that the expression for $\gamma$ has two terms: the first term is $\sim
v_{c}^{s_{1}-p/2}$, and the second term is $\sim v_{c}^{p/2-s_{1}}$. Since the
first particle is near-subcritical, it holds that $s_{1}<p/2$. From this fact,
we can deduce that the first term is dominant, and we have:%
\begin{equation}
\gamma\approx v_{c}^{-(p/2-s_{1})}.
\end{equation}

\item For the 1-st particle $v_{c}^{(1)}\rightarrow v_{t}^{(1)},$ while for
the second one $v_{c}^{(2)}\sim v_{e,t}^{(2)}.$ In this case for 1-st particle
condition (\ref{x_b_nsc}) holds and the behaviour of gamma factor may be
obtained by change $s_{1}\rightarrow r_{1}.$ Thus we have:%
\begin{equation}
\gamma\approx v_{c}^{-(p/2-r_{1})}.
\end{equation}

\qquad Note that taking limit $v_{c}^{(2)}\rightarrow v_{t}^{(2)}$ does not
change the relations $P_{2}\sim\mathcal{X}_{2}\sim\delta_{2}$. So this case
does not require special analysis.
\end{itemize}

\subsubsection{Both particles are NC, NUC or NOC}

According to Section \ref{iv_b}, the relations $\mathcal{X}_{1,2}\sim
P_{1,2}\sim\delta_{1,2}$ hold for both particles$.$ Substituting this in the
expression for $\gamma$ (\ref{gamma}), one has:%
\begin{equation}
\gamma\sim\frac{\delta_{1}\delta_{2}}{v_{c}^{p}}.
\end{equation}

Using that for NC, NUC, and NOC particles holds $\delta_{1,2}\sim v_{c}^{p/2}%
$, we can deduce that $\gamma=O(1)$.

Now, let us briefly formulate all the aforementioned conditions. To this end,
we introduce a parameter $d$ that characterizes how fast $\gamma$ changes as a
function of the collision point $v_{c}:\gamma\sim v_{c}^{-d}$. By analyzing
(\ref{gamma_1us_2_nsc}-\ref{gamma_1_cr_2nsc}) and (\ref{gamma_1_us_2_nc}%
-\ref{gamma_1_cr_2_nc}), we obtain the conditions summarized in Tables
\ref{gamma_delta_tab} and \ref{gamma_d_tab}.\begin{table}[ptb]%
\begin{tabular}
[c]{|c|c|c|c|}\hline
& First particle & Second particle & $d$\\\hline
1 & U or SC & NSC & $\left\vert s_{1}-s_{2}\right\vert $ or $\left\vert
s_{1}-r_{2}\right\vert $ if $\delta_{2}=-X_{s}^{(2)}v_{c}^{s}+B_{r}^{(2)}%
v_{c}^{r}$\\\hline
2 & C or UC & NSC & $\frac{p}{2}-s_{2}$ or $\frac{p}{2}-r_{2}$ if $\delta
_{2}=-X_{s}^{(2)}v_{c}^{s}+B_{r}^{(2)}v_{c}^{r}$\\\hline
3 & U or SC & NC, NUC or NOC & $\frac{p}{2}-s_{1}$\\\hline
4 & C or UC & NC, NUC or NOC & 0\\\hline
\end{tabular}
\caption{ Table showing behaviour of gamma factor for $v_{c}\sim v_{e,t}$ in a
case when 1-st particle is fine-tuned and 2-nd is near fine-tuned$.$ Here $d$
is defined by relation $\gamma\sim v_{c}^{-d}$}%
\label{gamma_delta_tab}%
\end{table}\begin{table}[ptb]%
\begin{tabular}
[c]{|c|c|c|c|}\hline
& First particle & Second particle & $d$\\\hline
1 & NSC & NSC & $|s_{1}-s_{2}|$\\
&  &  & $|r_{1}-s_{2}|$ if $\delta_{1}=-X_{s}^{(1)}v_{c}^{s}+B_{r}^{(1)}%
v_{c}^{r}$\\
&  &  & $|s_{1}-r_{2}|$ if $\delta_{2}=-X_{s}^{(2)}v_{c}^{s}+B_{r}^{(2)}%
v_{c}^{r}$\\
&  &  & $|r_{1}-r_{2}|$ if $\delta_{1,2}=-X_{s}^{(1,2)}v_{c}^{s}+B_{r}%
^{(1,2)}v_{c}^{r}$\\\hline
2 & NSC & NC, NUC or NOC & $\frac{p}{2}-s_{1}$\\\hline
3 & NC,NUC or NOC & NC, NUC or NOC & 0\\\hline
\end{tabular}
\caption{ Table showing behaviour of gamma factor for $v_{c}\sim v_{e,t}$ in a
case when both particles are near fine-tuned$.$ Here $d$ is defined by
relation $\gamma\sim v_{c}^{-d}$}%
\label{gamma_d_tab}%
\end{table}

Now, let us discuss the new scenarios that can be obtained. It is important to
note that for cases when $v_{c}\sim v_{e,t}$, the behavior of the gamma factor
is the same as for corresponding particles with $\delta=0$. The only
difference is the existence of near-overcritical particles, which do not have
any analog in the case of $\delta=0$, and the possibility of collisional
processes involving them. As we will show, such particles are very important
because they allow for high-energy collisions for non-extremal horizons.

There are also additional differences in cases when $v_{c}\rightarrow v_{t}$.
In these cases, the behavior of the gamma factor is described by different
expressions. For example, if two near-subcritical particles with $s_{1}=s_{2}$
participate in a collision, according to Table \ref{gamma_d_tab},
$\gamma=O(1)$. However, if we choose $\delta^{\prime}s$ for these particles in
such a way that $r_{1}\neq r_{2}$, then the gamma factor diverges. Thus, in
some cases, "fine-tuning" of the $\delta$'s makes previously forbidden BSW
effects possible.

\section{Behavior of acceleration: general approach\label{accel}}

Now we are going to analyze the forces acting on particles of different types.
In order to do this, at first we need to answer a question: in which frame do
we have to compute acceleration? Since the stationary frame is singular near
the horizon, we have to choose a frame that does not have this property. The
natural frame for this purpose is the one attached to a particle, known as the
FZAMO (free-falling \ zero angular momentum observer) frame. It is worth
noting that for near-fine-tuned particles, the radial velocity on the horizon
is not zero and the particle can cross the horizon, so we have to compute
acceleration in the FZAMO frame, unlike fine-tuned particles for which the
FZAMO frame is singular and acceleration has to be computed in the
non-singular OZAMO (orbital zero angular momentum observer) frame \cite{tz13}.
General definition of ZAMO and description of its properties is given in
\cite{72}.

To proceed further, let us use the expressions for acceleration in the tetrad
frame which are obtained in Appendix \ref{app_acel}.%
\begin{align}
a_{f}^{(t)}  &  =0,\\
a_{f}^{(r)}  &  =\frac{u^{r}}{\sqrt{\mathcal{X}^{2}-N^{2}}}(\partial
_{r}\mathcal{X+L}\partial_{r}\omega),\label{ar_fzamo}\\
a_{f}^{(\varphi)}  &  =\frac{1}{\sqrt{\mathcal{X}^{2}-N^{2}}}\frac{1}%
{\sqrt{g_{\varphi\varphi}}}\frac{\sqrt{A}}{N}\left[  (\mathcal{X}^{2}%
-N^{2})\partial_{r}\mathcal{L}-\mathcal{LX}(\partial_{r}\mathcal{X+L}%
\partial_{r}\omega)\right]  , \label{aphi_fzamo}%
\end{align}
where subscript $f$ means that the corresponding components of acceleration
are computed in the FZAMO frame. As one can see, there are only two non-zero
acceleration components and we are going to analyze at first $a_{f}^{(r)}$ to
understand the structure of acceleration and which terms are dominant for each
case discussed in Section \ref{diff_scales}. To this end, we use expansions
(\ref{an_exp}-\ref{om_exp}) and (\ref{X_us_2}) that give us%
\begin{equation}
a_{f}^{(r)}\approx\sigma\frac{P}{\sqrt{\mathcal{X}^{2}-N^{2}}}\sqrt
{\frac{A_{q}}{\kappa_{p}}}(X_{s}sv_{c}^{s+\frac{q-p}{2}-1}+L_{H}\omega
_{k}kv_{c}^{k+\frac{q-p}{2}-1}), \label{at}%
\end{equation}
where $\sigma=\pm1$ depending on whether the particle is outgoing or ingoing.
Our task is to consider different ranges of $v_{c}$ and find the behavior of acceleration.

\begin{itemize}
\item If $v_{c}\sim v_{e,t}$,%
\begin{equation}
a_{f}^{(r)}\approx(a_{f}^{(r)})_{m_{1}}\delta^{m_{1}},\text{ \ \ }%
a_{f}^{(\varphi)}\approx(a_{f}^{(\varphi)})_{m_{2}}\delta^{m_{2}}.
\label{a_expan_eb}%
\end{equation}

\item If $v_{c}\ll v_{e,t},$%
\begin{equation}
a_{f}^{(r)}\approx(a_{f}^{(r)})_{n_{1}}v_{c}^{n_{1}},\text{ \ \ }%
a_{f}^{(\varphi)}\approx(a_{f}^{(\varphi)})_{n_{2}}v_{c}^{n_{2}}.
\label{a_expan_es}%
\end{equation}

\item If there exists a turning point, in the limit $v_{c}\rightarrow v_{t}$
we have%
\begin{equation}
a_{f}^{(r)}\approx(a_{f}^{(r)})_{i_{1}}v_{c}^{i_{1}},\text{ \ \ }%
a_{f}^{(\varphi)}\approx(a_{f}^{(\varphi)})_{i_{2}}v_{c}^{i_{2}}.
\end{equation}

\end{itemize}

\bigskip

Our goal is to obtain relations between $m_{1},m_{2},n_{1},n_{2},i_{1},i_{2}$
and the type of a particle.

\bigskip

\section{Case $v_{c}\ll v_{e,t}\label{sec_eps<<1}$}

As we discussed in the analysis of the gamma factor (see Section
\ref{energy}), the case $v_{c}\ll v_{e,t}$ corresponds to usual particles.
However, in \cite{ov-zas-prd}, the acceleration for this case was analyzed
only qualitatively (the reason for this is that in \cite{ov-zas-prd}, the
authors were mainly focused on acceleration for fine-tuned particles, for
which the FZAMO frame is singular on the horizon). To fill this gap, we are
going to analyze acceleration of usual particles in the present work.

\subsection{General analysis of acceleration}

We start with the radial component of acceleration. First of all, we refer to
the fact that $P\sim\sqrt{\mathcal{X}^{2}-N^{2}}\sim\delta$ for $v_{c}\ll
v_{e,t}$ (to obtain this, one has to use (\ref{p_vc<<vt}) and
(\ref{sqrt_vc<<vt}). The dominant term $\mathcal{X}$ in these expressions has
the order $\delta$ on the horizon). Using this fact, we see that the prefactor
$\frac{P}{\sqrt{\mathcal{X}^{2}-N^{2}}}$ in (\ref{at}) is $=O(1)$. Thus, we
are left with the terms in brackets in (\ref{at}). One can see that the first
term is $\sim v_{c}^{s+\frac{q-p}{2}-1}$, and the second term is $\sim
v_{c}^{k+\frac{q-p}{2}-1}$. Comparing these terms with (\ref{a_expan_es}), we
get:%
\begin{equation}
n_{1}=\min(s,k)+\frac{q-p}{2}-1. \label{n1_eq_1}%
\end{equation}

However, this expression does not describe the most general case. It may
appear that the function $\omega=\omega_{H}$ is constant (that corresponds to
a static metric because there exists a corresponding coordinate transformation
$\widetilde{\varphi}=\varphi-\omega_{H}t$ that brings the metric to an
explicitly static form). In this case, the second term in (\ref{at}) is absent
and we obtain:%
\begin{equation}
n_{1}=s+\frac{q-p}{2}-1\text{ if }\omega=\omega_{H}. \label{n1_eq_2}%
\end{equation}

There is also a special case when the coefficients in the expansion of
$\mathcal{X}$ and $\mathcal{L}$ are such that several terms in the power
series (potentially divergent, generally speaking) in the expression for
acceleration cancel each other. Full cancellation happens, for example, for a
freely falling particle, provided $\mathcal{X}+\omega\mathcal{L}=\epsilon$,
where $\epsilon$ and $\mathcal{L}$ are constants that give us zero
acceleration. Then, $\partial_{r}(\mathcal{X}+\omega\mathcal{L})=0$ exactly.
In a more general case, we can consider%
\begin{equation}
\epsilon=\mathcal{X}+\omega\mathcal{L}=\epsilon_{0}+\text{terms of }%
v^{m}\text{ order, }m>k. \label{ener_spec_case}%
\end{equation}

Then, let us rewrite acceleration (\ref{ar_fzamo}) as a function of
$\epsilon:$%
\begin{equation}
a_{f}^{(r)}=\frac{u^{r}}{\sqrt{\mathcal{X}^{2}-N^{2}}}(\partial_{r}%
\epsilon\mathcal{-}\omega\partial_{r}\mathcal{L}).
\end{equation}

We have%
\begin{equation}
n_{1}=\min(m,b)+\frac{q-p}{2}-1\text{ in a case of (\ref{ener_spec_case}),}
\label{n1_eq_3}%
\end{equation}

(the quantity $b$ was defined in (\ref{L_exp})). However, we will not pay much
attention to this case in our further analysis.

Now let us invert all aforementioned conditions to obtain $s$ as a function of
$n_{1}$:%
\begin{equation}
s=\left\{
\begin{tabular}
[c]{l}%
$n_{1}+1+\frac{p-q}{2}\text{ if }0\leq n_{1}<k+\frac{q-p}{2}-1,$\\
$\text{may be any value }s\geq k\geq0\text{ if }n_{1}=k+\frac{q-p}{2}-1,$\\
$n_{1}+1+\frac{p-q}{2}\text{ for any }0\leq n_{1}\text{ if }\omega=\omega
_{H},$\\
$\text{may be any value in a case of (\ref{ener_spec_case}) (}n_{1}%
=\min(m,b)+\frac{q-p}{2}-1\text{).}$%
\end{tabular}
\ \ \ \ \ \ \right.  \label{s_on_n0}%
\end{equation}

Now, let us move to an angular component of acceleration. To analyze
finiteness of this component we note that it follows from (\ref{ar_fzamo}%
-\ref{aphi_fzamo}) and the fact that $u^{r}=\sigma\frac{\sqrt{A}}{N}P$ that
\begin{equation}
\mathcal{XL}a_{f}^{(r)}+\sigma P\sqrt{g_{\varphi\varphi}}a_{f}^{(\varphi
)}=\frac{\sqrt{A}}{N}P\sqrt{\mathcal{X}^{2}-N^{2}}\partial_{r}\mathcal{L}.
\label{ar_aphi_rel_1}%
\end{equation}

As in this section we are considering the near-horizon limit,
$\mathcal{X\approx}P\approx\sqrt{\mathcal{X}^{2}-N^{2}}\approx\delta.$ Thus we
can write in this limit%
\begin{equation}
a_{f}^{(\varphi)}+\sigma\mathcal{L}a_{f}^{(r)}=\delta\frac{\sqrt{A}}%
{N}\partial_{r}\mathcal{L}. \label{ar_aphi_rel}%
\end{equation}

We require that both $a_{f}^{(\varphi)}$ and $a_{f}^{(r)}$ be finite. From
(\ref{ar_aphi_rel}), it follows that if the acceleration components are
finite, then the left-hand side of (\ref{ar_aphi_rel}) is also finite. Because
the equality (\ref{ar_aphi_rel}) has to hold, the right-hand side also has to
be finite. If $q\geq p$, then the ratio $\frac{\sqrt{A}}{N}$ is finite on the
horizon, and the whole right-hand side is finite for any expansion of
$\mathcal{L}$. While if $q<p$, then the right-hand side behaves like%
\begin{equation}
\frac{\sqrt{A}}{N}\partial_{r}\mathcal{L\sim}v_{c}^{b+\frac{q-p}{2}-1},
\end{equation}

(here we used (\ref{L_exp})). This quantity is finite only for%

\begin{equation}
b\geq1+\frac{p-q}{2}. \label{b_cond}%
\end{equation}

Reversed statement thus can be easily proved by extraction of $a_{f}%
^{(\varphi)}$ $\ $from (\ref{ar_aphi_rel}):%

\[
a_{f}^{(\varphi)}=\delta\frac{\sqrt{A}}{N}\partial_{r}\mathcal{L-\sigma
L}a_{f}^{(r)}.
\]
If (\ref{b_cond}) holds and $a_{f}^{(r)}$ is finite, then $a_{f}^{(\varphi)}$
is also finite. This allows us to state a proposition:

\begin{proposition}
\label{propos_1}\bigskip If in case $v_{c}\ll v_{e,t}$ for some particle
$a_{f}^{(r)}$ is finite and condition (\ref{b_cond}) holds, then
$a_{f}^{(\varphi)}$ is also finite and vice versa.
\end{proposition}

\subsection{Behavior of acceleration for different types of particles}

\bigskip In this subsection we are going to analyze, which types of particles
are compatible with finite acceleration.

\subsubsection{Near-subcritical particles}

We start with near-subcritical particles. We would like to remind a reader
that the defining condition for them is $0<s<p/2$. Let us begin with the first
solution in (\ref{s_on_n0}). By substituting the range of $s:0<s<p/2$, we
obtain the corresponding range for $n_{1}:$ $\frac{q-p}{2}-1<n_{1}<\frac
{q-2}{2}$. Now, let us determine how this correlates with the condition for
the existence of the first solution in (\ref{s_on_n0}): $n_{1}\leq
k+\frac{q-p-2}{2}$. For $k<\frac{p}{2}$, the condition $n_{1}\leq
k+\frac{q-p-2}{2}$ is stronger than $n_{1}<\frac{q-2}{2}$ (it is also worth
noting that, according to the first solution in (\ref{s_on_n0}), $s<k$ in this
case). However, for $k\geq\frac{p}{2}$, the condition $n_{1}<\frac{q-2}{2}$
becomes stronger. The lower bound for $n_{1}$ remains the same for all
positive $k$. Therefore, we can conclude that for the first solution in
(\ref{s_on_n0})%
\begin{align}
\frac{q-p}{2}-1  &  <n_{1}\leq k+\frac{q-p}{2}-1\text{ if }s<k<\frac{p}%
{2},\label{n1_1_nearsub}\\
\frac{q-p}{2}-1  &  <n_{1}<\frac{q-2}{2}\text{ if }k\geq\frac{p}{2}.
\label{n1_1_2_nearsub}%
\end{align}

The second solution in (\ref{s_on_n0}) gives any $s$ that is greater than $k$
with $n_{1}$ fixed by (\ref{n1_eq_2}): $n_{1}=k+\frac{q-p}{2}-1.$ This value
is non-negative if
\begin{equation}
k\geq\frac{p-q}{2}+1. \label{k_cond_subcr}%
\end{equation}

As we are specifically considering near-subcritical particles with $0<s<p/2$,
this solution is possible only if $0<k\leq s<p/2$. This condition necessitates
that $k<p/2$. By combining this condition with (\ref{k_cond_subcr}), we can
deduce that $q>2$ for this solution to exist. In all other cases, this
solution does not exist. To summarize the above findings regarding the
existence of the second solution in (\ref{s_on_n0}), we have:%
\begin{align}
n_{1}  &  =k+\frac{q-p}{2}-1\text{ if }0<k\leq s<p/2,\label{n1_2_nearsub}\\
\text{Second solution in (\ref{s_on_n0}) is absent if }0  &  <s<k<p/2\text{ or
if }k\geq p/2. \label{n1_2_2_nearsub}%
\end{align}

If the function $\omega$ is constant (the third solution in (\ref{s_on_n0})),
we obtain the same range within which $n_{1}$ can change. It is the first
solution in (\ref{s_on_n0}): $\frac{q-p}{2}-1<n_{1}<\frac{q-2}{2}$.

By combining these facts, we can write:%
\begin{equation}
\text{For near-subcritical particles holds }\left\{
\begin{tabular}
[c]{l}%
$\frac{q-p}{2}-1<n_{1}\leq k+\frac{q-p}{2}-1\text{ if }k<\frac{p}{2},$\\
$\frac{q-p}{2}-1<n_{1}<\frac{q-2}{2}\text{ if }k\geq\frac{p}{2}\text{ or if
}\omega=\omega_{H}.$%
\end{tabular}
\ \ \text{ }\right.  \text{ } \label{n1_nearsub_fin}%
\end{equation}

\subsubsection{Near-critical and near-ultracritical particles}

In these cases, when $s=\frac{p}{2}$ (but in the case of near-ultracritical
particles, the additional condition (\ref{ultra_crit_prop}) must also hold),
despite the difference in $P$ functions, the acceleration for them is the same
(because acceleration depends only on $s$ which is the same for the considered
cases, see (\ref{s_on_n0})).

If we substitute $s=p/2$ into the first solution in (\ref{s_on_n0}), it will
give us $n_{1}=$ $\frac{q-2}{2}$, so $n_{1}$ is non-negative if $q\geq2$. The
first solution exists if $n_{1}<k+\frac{q-p}{2}-1$ that requires $k>\frac
{p}{2}$.

The second solution is true for all $s=p/2\geq k$ and, similarly to the case
of near-subcritical particles, gives a non-negative $n_{1}$ only if
$k\geq\frac{p-q}{2}+1$ (in this case, $n_{1}=$ $k+\frac{q-p}{2}-1$).

If the function $\omega$ is constant (the third solution in (\ref{s_on_n0})),
we get $n_{1}=$ $\frac{q-2}{2}$ without any additional limitations.

We can combine all the aforementioned solutions and write:%
\begin{equation}
\text{For NC and NUC particles holds }\left\{  \text{%
\begin{tabular}
[c]{l}%
$n_{1}=k+\frac{q-p-2}{2}\text{ if }k\leq\frac{p}{2},$\\
$n_{1}=\frac{q-2}{2}\text{ if }k>\frac{p}{2}$ or $\text{if }\omega=\omega
_{H}.$%
\end{tabular}
}\right.  \text{ } \label{n1_nearcr_fin}%
\end{equation}

\subsubsection{Near-overcritical particles}

For near-overcritical particles $s>\frac{p}{2}.$ The first solution in
(\ref{s_on_n0}) gives us $n_{1}>$ $\frac{q-2}{2}.$ As the first solution
exists for $n_{1}<k+\frac{q-p}{2}-1$ only, we have:%
\begin{equation}%
\begin{tabular}
[c]{l}%
$\frac{q-2}{2}<n_{1}<k+\frac{q-p}{2}-1\text{ if }k>\frac{p}{2}$\\
$\text{First solution is impossible if }k\leq\frac{p}{2}$%
\end{tabular}
\ \label{n1_1_1_nearov}%
\end{equation}

Substituting the inequalities $\frac{q-2}{2}<n_{1}<k+\frac{q-p}{2}-1$ back to
(\ref{s_on_n0}), one gets that the 1st solution in (\ref{s_on_n0}) requires
$p/2<s<k$.

The second solution in (\ref{s_on_n0}) allows all $s\geq k$ and is only
possible if $k\geq\frac{p-q}{2}+1$ (in which case $n_{1}=$ $k+\frac{q-p}{2}-1$).

The third solution simply gives us $\frac{q-2}{2}<n_{1}$ without any upper limitations.

Now let us join all the aforementioned solutions. Let us consider the case of
$k>p/2$. Depending on whether $s<k$ or $s\geq k$, we get different solutions.
In the case $s<k$, we have $\frac{q-2}{2}<n_{1}<k+\frac{q-p}{2}-1$ (see
(\ref{n1_1_1_nearov}) and the discussion after it), while in the case $s\geq
k$, we have $n_{1}=$ $k+\frac{q-p}{2}-1$. We can join these conditions and
write $\frac{q-2}{2}<n_{1}\leq k+\frac{q-p}{2}-1$ for any $k>p/2$,
independently of whether $s$ is greater or smaller than $k$. In the case
$k\leq p/2$, the situation is simpler: solution (\ref{n1_1_1_nearov}) is
absent, and we are left only with the second solution in (\ref{s_on_n0}):
$n_{1}=$ $k+\frac{q-p}{2}-1$. By joining all these cases, we can write:%
\begin{equation}
\text{For near-overcritical particles holds}\left\{  \text{%
\begin{tabular}
[c]{l}%
$\frac{q-2}{2}<n_{1}\leq k+\frac{q-p}{2}-1\text{ if }k>\frac{p}{2},$\\
$n_{1}=k+\frac{q-p}{2}-1$ if $k\leq p/2,$\\
$n_{1}>\frac{q-2}{2}\text{ if }\omega=\omega_{H}.$%
\end{tabular}
}\right.  \text{ } \label{n1_nearover_fin}%
\end{equation}

\subsection{Different $k$ regions}

To reformulate conditions obtained in previous subsections, first of all we
note that the existence of solutions is defined only by different values of
$k$. To simplify analysis further, we introduce different $k$ regions as we
did in Section VII.E in \cite{ov-zas-prd}.%
\begin{equation}%
\begin{tabular}
[c]{l}%
$0<k<\frac{p-q}{2}+1,\text{ \ \ region I,}$\\
$\frac{p-q}{2}+1\leq k<\frac{p+1-q/2}{2},\text{ \ \ region II,}$\\
$\frac{p+1-q/2}{2}\leq k<\frac{p}{2},\text{ \ \ \ region III,}$\\
$k\geq\frac{p}{2},\text{ \ \ region IV.}$%
\end{tabular}
\ \ \ \label{k_reg_q>2}%
\end{equation}
\qquad

As was discussed in \cite{ov-zas-prd}, all these regions exist and do not
intersect if $q>2.$ However if $q\leq2,$ then classification in these cases
has to be introduced differently:
\begin{equation}%
\begin{tabular}
[c]{l}%
$0<k<\frac{p-q}{2}+1,\text{ \ \ region I,}$\\
$k\geq\frac{p-q}{2}+1,\text{ \ \ region IV.}$%
\end{tabular}
\ \ \label{k_reg_q<2}%
\end{equation}

Let us start our analysis with the stationary metric (where $\omega$ is
non-constant). We see that in region I $n_{1}$ is negative for any type of
particle because in this case it follows from (\ref{n1_nearsub_fin}),
(\ref{n1_nearcr_fin}) and (\ref{n1_nearover_fin}) that for any type of
particle, $n_{1}$ is either limited by the value $k+\frac{q-p-2}{2}$ or is
equal to it. However, as $k+\frac{q-p-2}{2}$ is negative in region I, $n_{1}$
is also negative).

In regions II and III (where $\frac{p-q}{2}+1\leq k<\frac{p}{2}$) we see that
for near-subcritical particles, we have to choose the 1st solution in
(\ref{n1_nearsub_fin}) that gives us $\frac{q-p-2}{2}<n_{1}\leq k+\frac
{q-p-2}{2}$. For near-critical and near-ultracritical particles, we take the
1st solution in (\ref{n1_nearcr_fin}). For near-overcritical particles, we
take the 2-nd solution in (\ref{n1_nearover_fin}) that gives us $n_{1}%
=k+\frac{q-p}{2}-1$ for all these particles.

In region IV (where $k\geq\frac{p}{2}$), for near-subcritical particles, we
use the 2nd solution in (\ref{n1_nearsub_fin}) that gives us $\frac{q-p-2}%
{2}<n_{1}<\frac{q-2}{2}$. To include non-negative values of $n_{1}$ in this
region, we have to require $q>2$. Otherwise, accelerations for
near-subcritical particles diverge in region IV.

For near-critical and near-ultracritical particles we use either 1-st solution
in (\ref{n1_nearcr_fin}) (if $k=p/2$) that entails $n_{1}=\frac{q-2}{2}$ or we
use 2-nd solution in (\ref{n1_nearcr_fin}) that leads to the same value
$n_{1}=\frac{q-2}{2}$. The acceleration in these cases is non-negative only if
$q\geq2$.

For near-overcritical particles, we use either the 1st solution in
(\ref{n1_nearover_fin}) that gives us $\frac{q-2}{2}<n_{1}<k+\frac{q-p}{2}-1$
(this solution is true if $\frac{p}{2}<s<k$), or we use the 2nd solution in
(\ref{n1_nearover_fin}) that gives us $n_{1}=k+\frac{q-p}{2}-1$ (if $s\geq k.$
This solution in region IV is presented only if $k=\frac{p}{2}$). These two
solutions in region IV can be joined that gives us for near-overcritical
particles $\frac{q-2}{2}<n_{1}\leq k+\frac{q-p}{2}-1$ (independently of
whether $s$ or $k$ is greater). This condition does not have further
limitations if $q>2$.

If $q\leq2$, the lower bound for $n_{1}$ is negative. The upper bound for
$n_{1}$ (that is the same as in the case $q>2$) is positive, with the
reservation about redefinition of $k$ regions in this case (see
(\ref{k_reg_q<2})).

For a constant $\omega$, we observe that near-subcritical, near-critical, and
near-ultracritical particles experience the same accelerations as non-constant
$\omega$ in region IV (one can refer to 2nd solutions in (\ref{n1_nearsub_fin}%
) and (\ref{n1_nearcr_fin}). The only difference occurs with near-overcritical
particles, where in the case of constant $\omega,$ $n_{1}>\frac{q-2}{2}$ (see
3rd solution (\ref{n1_nearover_fin})).

We summarize all the aforementioned results in Tables \ref{tab_1},\ref{tab_2},
and \ref{tab_3}.

\begin{table}[ptb]%
\begin{tabular}
[c]{|c|c|c|c|c|}\hline
& $k$ region & $n_{1}$ range & $s$ & Type of trajectory\\\hline
\multicolumn{5}{|c|}{Stationary metric}\\\hline
1 & I & \multicolumn{3}{|c|}{For any type of trajectory $n_{1}$ is negative
(forces diverge)}\\\hline
2 & II and III & $\max\left(  0,\frac{q-p}{2}-1\right)  <n_{1}\leq
k+\frac{q-p}{2}-1$ & 1-st and 2-nd in (\ref{s_on_n0}) & NSC\\\cline{3-5}
&  & $n_{1}=k+\frac{q-p}{2}-1$ & 2-nd in (\ref{s_on_n0}) & NC and
NOC\\\cline{3-5}
&  & $n_{1}=k+\frac{q-p}{2}-1$ and (\ref{ultra_crit_prop}) & 2-nd in
(\ref{s_on_n0}) & NUC\\\hline
3 & IV & $\max\left(  0,\frac{q-p}{2}-1\right)  <n_{1}<\frac{q-2}{2}$ & 1-st
in (\ref{s_on_n0}) & NSC\\\cline{3-5}
&  & $n_{1}=\frac{q-2}{2}$ & 1-st in (\ref{s_on_n0}) & NC\\\cline{3-5}
&  & $n_{1}=\frac{q-2}{2}$ and (\ref{ultra_crit_prop}) & 1-st in
(\ref{s_on_n0}) & NUC\\\cline{3-5}
&  & $\frac{q-2}{2}<n_{1}\leq k+\frac{q-p}{2}-1$ & 1-st and 2-nd in
(\ref{s_on_n0}) & NOC\\\hline
\multicolumn{5}{|c|}{Static metric}\\\hline
4 & $k=0$ & \multicolumn{2}{|c|}{Same results as in IV for stationary metric}
& NSC, NC and NUC\\\cline{3-5}
&  & $n_{1}>\frac{q-2}{2}$ & 3-rd in (\ref{s_on_n0}) & NOC\\\hline
\end{tabular}
\caption{ Classification of near-horizon trajectories for different k regions
for $q>2$ (ultraextremal horizon). Abbreviations mean: NSC-mear-subcritical,
NC-near-critical, NUC-near-ultracritical, NOC-near-overcritical. The fourth
solution in (\ref{s_on_n0}) is not presented in this table. }%
\label{tab_1}%
\end{table}

\begin{table}[ptb]%
\begin{tabular}
[c]{|c|c|c|c|c|}\hline
& $k$ region & $n_{1}$ range & $s$ & Type of trajectory\\\hline
\multicolumn{5}{|c|}{Stationary metric}\\\hline
1 & I & \multicolumn{3}{|c|}{For any type of trajectory $n_{1}$ is negative
(forces diferge)}\\\hline
2 & IV & $n_{1}=0$ & 1-st in (\ref{s_on_n0}) & NC\\\cline{3-5}
&  & $n_{1}=0$ and (\ref{ultra_crit_prop}) & 1-st in (\ref{s_on_n0}) &
NUC\\\cline{3-5}
&  & $0\leq n_{1}\leq k+\frac{q-p}{2}-1$ & 1-st and 2-nd in (\ref{s_on_n0}) &
NOC\\\hline
\multicolumn{5}{|c|}{Static metric}\\\hline
3 & $k=0$ & \multicolumn{2}{|c|}{Same results as in IV for stationary metric}
& NC and NUC\\\cline{3-5}
&  & $n_{1}>0$ & 3-rd in (\ref{s_on_n0}) & NOC\\\hline
\end{tabular}
\caption{ Classification of near-horizon trajectories for different k regions
for $q=2$ (ultraextremal horizon). Abbreviations mean: NSC-mear-subcritical,
NC-near-critical, NUC-near-ultracritical, NOC-near-overcritical. The fourth
solution in (\ref{s_on_n0}) is not presented in this table. }%
\label{tab_2}%
\end{table}

\begin{table}[ptb]%
\begin{tabular}
[c]{|c|c|c|c|c|}\hline
& $k$ region & $n_{1}$ range & $s$ & Type of trajectory\\\hline
\multicolumn{5}{|c|}{Stationary metric}\\\hline
1 & I & \multicolumn{3}{|c|}{For any type of trajectory $n_{1}$ is negative
(forces diverge)}\\\hline
2 & IV & $0<n_{1}\leq k+\frac{q-p}{2}-1$ & 1-st and 2-nd in (\ref{s_on_n0}) &
NOC\\\hline
\multicolumn{5}{|c|}{Static metric}\\\hline
3 & $k=0$ & $0<n_{1}$ & 3-rd in (\ref{s_on_n0}) & NOC\\\hline
\end{tabular}
\caption{ Classification of near-horizon trajectories for different k regions
for $q<2$ (ultraextremal horizon). Abbreviations mean: NSC-mear-subcritical,
NC-near-critical, NUC-near-ultracritical, NOC-near-overcritical. The fourth
solution in (\ref{s_on_n0}) is not presented in this table. }%
\label{tab_3}%
\end{table}

Also note that results obtained in this section include also the case of usual
particles (with arbitrary $\delta=O(1)$).

\section{Case $v_{c}\sim v_{e,t}$}

As we already mentioned, in this case in the main approximation acceleration reads%

\begin{equation}
a_{f}^{(r)}\approx(a_{f}^{(r)})_{m_{1}}v_{c}^{m_{1}},\text{ \ \ }%
a_{f}^{(\varphi)}\approx(a_{f}^{(\varphi)})_{m_{2}}v_{c}^{m_{2}}.
\end{equation}

Our task is to find $m_{1}$ and $m_{2}$ depending on the particle type. We
begin with the radial component of acceleration, $a_{f}^{(r)}$. At first, let
us discuss the prefactor in the expression for (\ref{at}). As stated in
Section \ref{iv_b}, for all types of particles, $P\sim\sqrt{\mathcal{X}%
^{2}-N^{2}}\sim\delta$ (note that this applies to the $v_{c}\sim v_{e,t}$
case, but not to the limit $v_{c}\rightarrow v_{t}$). Thus, the prefactor
$\frac{P}{\sqrt{\mathcal{X}^{2}-N^{2}}}=O(1)$ and can be ignored in this analysis.

In the expression for $a_{f}^{(r)}$ (\ref{at}), there are two terms dependent
on $v_{c}$: the first term is $\sim v_{c}^{s+\frac{q-p}{2}-1}$, and the second
term is $\sim v_{c}^{k+\frac{q-p}{2}-1}$. Therefore, the analysis of the
radial component of acceleration is the same as for $v_{c}\ll v_{e,t}$, and we
can use the results obtained in the previous section.

The only differences are related to the angular component of acceleration. In
this case, $\mathcal{X}\sim\sqrt{\mathcal{X}^{2}-N^{2}}\sim P\sim\delta$, and
we can use equation (\ref{ar_aphi_rel}):%
\begin{equation}
a_{f}^{(\varphi)}+\sigma\mathcal{L}a_{f}^{(r)}\sim\delta\frac{\sqrt{A}}%
{N}\partial_{r}\mathcal{L}.
\end{equation}

Relying on this condition, let us determine when $a_{f}^{(\varphi)}$ is finite
. The right-hand side (RHS) of this condition is of the order of $\delta
v_{c}^{b+\frac{q-p}{2}-1}$. It is evident that if condition (\ref{b_cond}) is
satisfied, then the RHS is also finite. Therefore, we can conclude that for
the range $v_{c}\sim v_{e,t}$, if $a_{f}^{(r)}$ is finite and condition
(\ref{b_cond}) is met, then $a_{f}^{(\varphi)}$ is also finite (see
Proposition \ref{propos_1}).

Considering that the expression for $n_{1}$ is the same for both cases
$v_{c}\sim v_{e,t}$ and $v_{c}\ll v_{e,t}$, we can write:

\begin{proposition}
If acceleration is finite for $v_{c}\ll v_{e,t}$, it is finite for $v_{c}\sim
v_{e,t}.$
\end{proposition}

This proposition is very useful because we can use the results already
obtained for $v_{c}\ll v_{e,t}.$

\section{Case $v_{c}\rightarrow v_{t}$}

If a particle is such that a turning point exists, the quantity $P$ tends to
zero when this point is \ approached. However, for different types of
particles, it tends to zero at different rates. Therefore, we will consider
each of these cases separately.

\subsection{Near-subcritical particles}

For such particles we use (\ref{x_b_nsc}-\ref{sqrt_nsc}) and see that in the
leading order%
\begin{equation}
P\approx\sqrt{\mathcal{X}^{2}-N^{2}}\approx\mathcal{X\approx}B_{r}v_{c}^{r}.
\end{equation}

Thus prefactor $\frac{P}{\sqrt{\mathcal{X}^{2}-N^{2}}}$ in expression
(\ref{at}) is $=O(1).$ This means that for such particles acceleration for
$v_{c}\rightarrow v_{t}$ case is given by the same expressions as for the case
of $v_{c}\sim v_{t}.$

\subsection{Near-critical, near-ultracritical and near-overcritical particles}

All these cases are similar in the sense that in the leading order (see
(\ref{x_nc_cond}), (\ref{delta_spec_cond_cr}), (\ref{p_b_cr}),
(\ref{delta_spec_cond_ucr}), (\ref{p_b_ultra}), (\ref{delta_spec_cond_ocr})
and (\ref{p_b_over})) we can write:%
\begin{equation}
\mathcal{X}\sim\sqrt{\mathcal{X}^{2}-N^{2}}\sim v_{c}^{p/2},\text{ \ \ }P\sim
v_{c}^{p/4+r/2}.
\end{equation}

Now, let us analyze the radial component of acceleration. It can be observed
that the prefactor $\frac{P}{\sqrt{\mathcal{X}^{2}-N^{2}}}$ in the expression
(\ref{at}) is approximately $\sim v_{c}^{(\frac{r}{2}-\frac{p}{4})}$. Note
that, by definition, $r>p/2$, so this prefactor does not diverge. The
structure of the terms in brackets is the same as in the previously analyzed
cases of $v_{c}\ll v_{e,t}$ and $v_{c}\sim v_{e,t}$. Therefore, we can
conclude that if the acceleration is finite in these two cases, it will also
be finite in the case of $v_{c}\rightarrow v_{t}$.

Next, we analyze the angular component (this analysis is independent of the
type of particle). To this end, we will use the expression (\ref{aphi_fzamo}):%
\begin{equation}
a_{f}^{(\varphi)}=\frac{1}{\sqrt{\mathcal{X}^{2}-N^{2}}}\frac{1}%
{\sqrt{g_{\varphi\varphi}}}\frac{\sqrt{A}}{N}\left[  (\mathcal{X}^{2}%
-N^{2})\partial_{r}\mathcal{L}-\mathcal{LX}(\partial_{r}\mathcal{X+L}%
\partial_{r}\omega)\right]  . \label{af_phi}%
\end{equation}

Using that in the limit $v_{c}\rightarrow v_{t}$ quantities $\sqrt
{\mathcal{X}^{2}-N^{2}}\sim\mathcal{X}$ are of order of $\delta$ (this is true
because $v_{t}$ is defined by condition $\mathcal{X}^{2}=(1+\frac
{\mathcal{L}^{2}}{g_{\varphi\varphi}})N^{2},$ while $\mathcal{X}$ and $N$
remain $\sim\delta$), and using expansions (\ref{an_exp}) and (\ref{L_exp}),
we have three terms in (\ref{af_phi}). The first term is $\sim\delta
v_{c}^{b+\frac{q-p}{2}-1}$, the second term is $\sim v_{c}^{s+\frac{q-p}{2}%
-1}$, and the third term is $\sim v_{c}^{k+\frac{q-p}{2}-1}$. The second and
third terms are finite if $\min(s,k)+\frac{q-p}{2}-1$ is non-negative. Note
that this expression is the same as $n_{1}$ in the case $v_{c}\ll v_{e,t}$
given by (\ref{n1_eq_1}). So, if the radial component of the acceleration is
finite in the case $v_{c}\ll v_{e,t}$, then the second and third terms are
also finite. The first term is finite if (\ref{b_cond}) is satisfied that is a
necessary condition for the finiteness of the angular component of the
acceleration in the case $v_{c}\ll v_{e,t}$. Therefore, we can state the proposition:

\begin{proposition}
If acceleration is finite for $v_{c}\ll v_{e,t}$, it is finite for
$v_{c}\rightarrow v_{t}.$
\end{proposition}

\section{Results for different types of horizons\label{sec_dif_types}}

In this section we formulate briefly the results obtained for different
particles near different types of horizons.\newline To this end, we analyze
accelerations for $v_{c}\ll v_{e,t}$, because, as we showed above (see
Propositions 1, 2 and 3), if acceleration is finite in this range, it will be
finite in all other ranges.

\subsubsection{Non-extremal horizons}

Non-extremal horizons are such that $q=p=1$ (for explanation, why we use such
a definition, and other properties of non-extremal horizons see
\cite{ov-zas-gc}). All results corresponding to this case may be found in
Table \ref{tab_3}. One can see that in this case (unlike that of fine-tuned
particle) it is possible to have particles that experience a finite force near
the horizon: such particles are near-overcritical ones in range IV (or for
static spacetimes). This correlates with the result obtained in \cite{zas20}
where such a situation was demonstrated explicitly for the Schwarzchild
space-time. According to our classification, such particles are
near-overcritical ones that, as we already discussed it above, do not have a
fine-tuned analogue.

\subsubsection{Extremal horizons}

Extremal horizons are such that $q=2$ and $p>q.$ All results corresponding to
this case may be found in Table \ref{tab_2}. From this Table, we see that
finite forces act only on near-critical, near-ultracritical and
near-overcritical particles in region IV or in static metric.

\subsubsection{Ultraextremal horizon}

Ultraextremal horizons are such that $q>2.$ All results corresponding to this
case may be found in Table \ref{tab_1}. One can see that in this case a force
is finite for all types of particles if $k$ is in regions II, III or IV or if
space-time is static.

We summarize all the aforementioned results in the Table \ref{tab_6}.

\begin{table}[ptb]%
\begin{tabular}
[c]{|c|c|c|c|}\hline
& Type of horizon & Type of trajectory & Region of $k$\\\hline
1 & Non-extremal & NOC & IV or static\\\hline
2 & Extremal & NC, NUC and NOC & IV or static\\\hline
3 & Ultraextremal & NSC, NC, NUC and NOC & II, III, IV or static\\\hline
\end{tabular}
\caption{ Classification of cases when forces are finite for different types
of horizons and trajectories.}%
\label{tab_6}%
\end{table}

\section{Particles with finite proper time: is kinematic censorship
preserved?}

In this section we are going to probe the so-called principle of kinematic
censorship (KC). It excludes unphysical situations in which the energy
released in a collisional event in a regular system is infinite in a literal
sense \cite{cens}. Something should prevent such an event. For example, the
proper time to the horizon can diverge \cite{ted}, so collision with infinite
$E_{c.m.}$ never occurs, although it can be made as large as one likes. What
happens if the corresponding proper time is finite and $E_{c.m.}$ is infinite?
In Sec. IX of our previous work \cite{ov-zas-prd} we showed that for
fine-tuned particles this is possible in two cases only: either 1) a force,
acting on such particles is infinite or 2) or the horizon fails to be regular.
As the system becomes singular in a geometrical or dynamic sense, the KC does
not apply to it and no contradiction arises. Now, we are going to prove that
this principle is preserved for near-fine-tuned particles. However, this
requires consideration of one more factor: 3) an interval within which motion
leading to diverging $E_{c.m.}$ is allowed, shrinks to the point. This makes
the scenario degenerate and unphysical.

Let us prove the corresponding theorem.

\begin{theorem}
If for a near-subcritical, near-critical or near-ultracritical particle a
proper time needed to reach the horizon is finite for all possible relations
between $v_{c}$ and $v_{t,e}$, then, at least one of three aforementioned
conditions is fulfilled
\end{theorem}

\begin{proof}
We start with an analysis of near-subcritical, near-critical and
near-ultracritical particles in the $v_{c}\gg v_{e,t}$ case. As was discussed
in the Section \ref{sec_4_a}, near-fine-tuned particles in this range behave
in the same way as corresponding fine-tuned ones. Corresponding part of the
theorem (requiring the horizon's regularity) was proven for such particles in
\cite{ov-zas-prd}.

For cases $v_{c}\ll v_{e,t}$ and $v_{c}\sim v_{e,t}$ proof is different. We
will conduct it by considering different types of particles separately.

Let us start with near-subcritical particles. As one can see from the Appendix
\ref{sec_prop_time}, a proper time for them is finite if $q<p+2$ (in the range
$v_{c}\ll v_{e,t}$) and $s<\frac{p-q}{2}+1$ (in the range $v_{c}\sim v_{e,t}%
$). As we require the proper time to be finite for all relations between
$v_{c}$ and $v_{e,t}$, this means that both these conditions have to hold.

Now let us analyze an acceleration for these particles. We will focus only on
the acceleration for $v_{c}\ll v_{e,t}$ because, as we showed above, if
acceleration is finite in this range it will be finite in other ranges.
Finiteness of the acceleration is defined by the sign of $n_{1}$ which is
given by (\ref{n1_eq_1}):%
\begin{equation}
n_{1}=\min(s,k)+\frac{q-p}{2}-1. \label{n1}%
\end{equation}

First of all, we note that from defining property of $\min(s,k)$ it follows
that $\min(s,k)\leq s,$ and using condition $s<\frac{p-q}{2}+1$ we have
$\min(s,k)<\frac{p-q}{2}+1.$ Property $q<p+2$ tells us that $\frac{p-q}{2}-1$
is positive value, so $n_{1}=\min(s,k)-(\frac{p-q}{2}+1)$ is negative that
means that a force diverges.

Now let us consider near-critical and near-ultracritical particles. For them,
$s=p/2$. As one can note from Appendix \ref{sec_prop_time}, a proper time for
these particles is finite if $q<2.$ As we mentioned, $\min(s,k)\leq s$ that
gives $\min(s,k)\leq p/2.$ Thus we see that $n_{1}$ is given by:
\begin{equation}
n_{1}=\left(  \min(s,k)-\frac{p}{2}\right)  +\left(  \frac{q}{2}-1\right)  .
\end{equation}

The first bracket is $\leq0$ because $\min(s,k)\leq p/2$), the second one is
$<0$ because $q<2$. This makes $n_{1}$ negative and makes a force divergent.

Now let us consider near-overcritical particles. First of all, we note that
according to the result we obtained in the Section \ref{sec_vt_noc}, such
particles always have a turning point and a particle may move only in the
$[0,v_{t}]$ interval. Thus the case $v_{c}\gg v_{e,t}$ is now impossible, so
conclusions from the theorem in the Section IX of \cite{ov-zas-prd} simply
cannot be applied to them.

Then we focus on the $v_{c}\ll v_{e,t}$ and $v_{c}\sim v_{e,t}$ cases. The
conditions of finiteness of the proper time for such particles are the same as
for near-critical and near-ultracritical particles, namely $q<p+2$ and $q<2$
(these conditions, obviously, can be merged to give $q<2$). However, such
particles have $s>p/2$ that causes the main difference between them and the
case of near-critical and near-ultracritical particles. Indeed, the expression
for $n_{1}$ is given by (\ref{n1}). This quantity may be made non-negative if
one chooses $s\geq\frac{p-q}{2}+1$ and $k\geq\frac{p-q}{2}+1$ (or if metric is
static). As $q<2$, the quantity $\frac{p-q}{2}+1>\frac{p}{2}.$ Because of
this, we obtain that $k>\frac{p}{2}$ (that corresponds to region IV) or the
metric is static.

Thus for NOC particles \ the proper time is finite, acceleration is also
finite. Meanwhile, an infinite $E_{c.m.}$ is possible for scenarios in which
such a particle collides with a subcritical, or near-subcritical particle
\ see Table II (line 3) and Table III (line 2). On the first glance, the KC
principle is violated. However, this is not so. Let us remind a reader that,
say, for nonextremal black holes, the BSW effect is possible with small but
nonzero $\delta$ but in such a way that a corresponding near-overcritical
particle moves in a very narrow strip. This was shown in \cite{gr-pav} for the
Kerr metric. For a more general rotating axially symmetric black hole, the
corresponding allowed interval of $N$ is proportional to $\delta$ according to
eq. (18) of Ref. \cite{prd}. In the limit $\delta\rightarrow0$ this interval
shrinks to a point and the process of collision (as well as motion of such a
particle) loses its sense. The KC is not violated since it cannot simply be
applied to a system.

The similar situation happens for NOC particle in a more general case. Indeed,
if $\delta\rightarrow0$, the corresponding $v_{t}\rightarrow0$ according to
(\ref{vt_overcr}) and the allowed coordinate interval $0\leq v_{c}\leq v_{t}$
shrinks to the point. Moreover, this is valid for the proper distance
\begin{equation}
l=\int_{v_{c}}^{v_{t}}\frac{dr}{\sqrt{A}}\text{.}%
\end{equation}
Indeed, as in the case under discussion $\sqrt{A}\sim v^{q/2}$ with $q<2$, the
integral converges and, as the limits of integration shrink, so does $\tau$.

This completes the proof of the theorem.
\end{proof}

\section{Varying ranges of motion of particles\label{ranges}}

Our previous analysis primarily focused on investigating the properties of the
collisional process for a given particle. Although we classified all possible
cases where the BSW phenomenon is possible with forces acting on particles,
there are still several problems with this approach. As is shown in Section
\ref{sec_class_of_part}, not all particles that potentially participate in the
BSW effect can reach freely the collision point from infinity. This means that
such particles can only be created in a narrow region disconnected from
infinity, either due to a quantum creation process in this region or through
multiple scattering \cite{gr-pav}. We consider these scenarios to be exotic
and put them aside. Instead, we are interested in the question of whether it
is possible to change the ranges of motion of a particle by the action of
force in such a way that the particle could reach infinity. To answer this
question, we need to analyze the force acting on a particle for different
types of particles.

\subsection{Near-subcritical particle}

We begin our analysis with near-subcritical particles (as we will show later,
the situation is similar for other particles). For near-subcritical particles,
using equation (\ref{fit_cond}), we find that the particle can reach infinity
only if $\delta>0$ and $X_{s}>0$. Therefore, our task is to determine the
conditions on a force under which we can have $X_{s}>0$. To do this, we
revisit the expression for the radial component of the acceleration
(\ref{at}):%
\begin{equation}
a_{f}^{(r)}\approx\sigma\frac{P}{\sqrt{\mathcal{X}^{2}-N^{2}}}\sqrt
{\frac{A_{q}}{\kappa_{p}}}(X_{s}sv_{c}^{s+\frac{q-p}{2}-1}+L_{H}\omega
_{k}kv_{c}^{k+\frac{q-p}{2}-1}).
\end{equation}

As we are considering near-subcritical particle, for which $s<p/2$ and
$\mathcal{X\ll}N,$ we have that $P\approx\sqrt{\mathcal{X}^{2}-N^{2}}%
\approx|\mathcal{X}|$ (for a proof see Section \ref{iv_b}), thus $\frac
{P}{\sqrt{\mathcal{X}^{2}-N^{2}}}\approx1.$ This allows us to write:%
\begin{equation}
a_{f}^{(r)}\approx\sigma\sqrt{\frac{A_{q}}{\kappa_{p}}}(X_{s}sv_{c}%
^{s+\frac{q-p}{2}-1}+L_{H}\omega_{k}kv_{c}^{k+\frac{q-p}{2}-1}).
\label{ar_x_nsc}%
\end{equation}

As we have previously discussed, there are two terms here and, depending on
whether $s>k,$ $s=k$ or $s<k$ , we obtain different behavior. Our task is to
consider these different cases and determine how the coefficients in the
expansions for the acceleration (\ref{a_expan_eb}) and (\ref{a_expan_es}) in
different ranges of the coordinate $v_{c}$ are related to $X_{s}$. This will
allow us to determine whether we can control the ranges of motion of the
particle or not. It is important to note that in this investigation, we will
only analyze the cases $v_{c}\ll v_{e,t}$ and $v_{c}\sim v_{e,t}$. The case
$v_{c}\gg v_{e,t}$ will not be investigated here because, as we have already
discussed, at such distances, the particle can be considered as fine-tuned
(with negligible influence of the parameter $\delta$) and we return to the
pure BSW effect.

\begin{itemize}
\item $s<k.$

\qquad In this case the 1-st term in (\ref{ar_x_nsc}) is dominant and we have
$a_{f}^{(r)}\approx\sigma\sqrt{\frac{A_{q}}{\kappa_{p}}}X_{s}sv_{c}%
^{s+\frac{q-p}{2}-1}.$ Substituting here (\ref{a_expan_eb}) and
(\ref{a_expan_es}) one obtains:%
\begin{equation}
(a_{f}^{(r)})_{n_{1}}=(a_{f}^{(r)})_{m_{1}}=\sigma\sqrt{\frac{A_{q}}%
{\kappa_{p}}}sX_{s},\text{ \ \ }n_{1}=m_{1}=s+\frac{q-p}{2}-1.
\end{equation}

\qquad This gives us in this case%
\begin{equation}
X_{s}=\sigma\sqrt{\frac{\kappa_{p}}{A_{q}}}(a_{f}^{(r)})_{m_{1}}=\sigma
\sqrt{\frac{\kappa_{p}}{A_{q}}}(a_{f}^{(r)})_{n_{1}}. \label{xs_nsc_s<k}%
\end{equation}

\qquad Therefore, we observe that the sign of $X_{s}$ and thus the possibility
of reaching infinity is determined by the sign of the acceleration. For
ingoing particles ($\sigma=-1$), it is only possible if the force is negative
(attractive). This case also corresponds to static spacetimes.

\item $s=k$

\qquad In this case both terms in (\ref{ar_x_nsc}) are comparable and we have
$a_{f}^{(r)}\approx\sigma\sqrt{\frac{A_{q}}{\kappa_{p}}}(X_{s}sv_{c}%
^{s+\frac{q-p}{2}-1}+L_{H}\omega_{k}kv_{c}^{k+\frac{q-p}{2}-1}).$ \qquad

\qquad Substituting here (\ref{a_expan_eb}) and (\ref{a_expan_es}) one
obtains:%
\begin{equation}
(a_{f}^{(r)})_{n_{1}}=(a_{f}^{(r)})_{m_{1}}=\sigma\sqrt{\frac{A_{q}}%
{\kappa_{p}}}(sX_{s}+L_{H}\omega_{k}k),\text{ \ \ }n_{1}=m_{1}=k+\frac{q-p}%
{2}-1.
\end{equation}

\qquad Reverting it, one gets%
\begin{equation}
X_{s}=\frac{1}{s}\left(  \sqrt{\frac{\kappa_{p}}{A_{q}}}\frac{(a_{f}%
^{(r)})_{n_{1}}}{\sigma}-L_{H}\omega_{k}k\right)  . \label{xs_nsc_s=k}%
\end{equation}

\qquad One can observe that the sign of this expression can be controlled by
the choice of the proper values of acceleration. For example, if the particle
is ingoing and we want to have a positive $X_{s}$, one must have an
acceleration that satisfies the condition $(a_{f}^{(r)})_{n_{1}}<-\sqrt
{\frac{A_{q}}{\kappa_{p}}}L_{H}\omega_{k}k.$

\item $s>k$

\qquad In this case, the second term in equation (\ref{ar_x_nsc}) is dominant,
and we have $a_{f}^{(r)}\approx\sigma\sqrt{\frac{A_{q}}{\kappa_{p}}}%
L_{H}\omega_{k}kv_{c}^{k+\frac{q-p}{2}-1}$. As one can see, in this case, the
acceleration is independent of $X_{s},$ so $X_{s}$ cannot be controlled by
leading terms in the external force.

\qquad Therefore, we observe that the range of motion of a near-subcritical
particle can be controlled by force if $s\leq k$. For these cases, $n_{1}$ is
defined by the first and the second solutions in equation (\ref{s_on_n0}).
\end{itemize}

\subsection{Near-critical, near-ultracritical and near-overcritical particles}

As far as the behavior of acceleration is concerned, all these cases are
similar. Let us start with the case $v_{c}\ll v_{e,t}$. Then, as is shown in
Section \ref{iv_c}, for any particle, $P\approx\sqrt{\mathcal{X}^{2}-N^{2}%
}\approx\delta$ that gives us $\frac{P}{\sqrt{\mathcal{X}^{2}-N^{2}}}\approx
1$. Thus, the expression for acceleration (\ref{at}) is the same as for the
case of a near-subcritical particle (\ref{ar_x_nsc}), so the analysis is the
same in those cases. In the case $v_{c}\sim v_{e,t}$, the quantity $\frac
{P}{\sqrt{\mathcal{X}^{2}-N^{2}}}$ differs from 1. Let us consider the Taylor
expansion for this quantity.
\begin{equation}
\frac{P}{\sqrt{\mathcal{X}^{2}-N^{2}}}\approx1+d_{p/2}v_{c}^{p/2},
\end{equation}

where $d_{p/2}$ is some coefficient. Substituting this to (\ref{at}) one gets:%
\begin{equation}
a_{f}^{(r)}\approx\sigma(1+d_{p/2}v_{c}^{p/2})\sqrt{\frac{A_{q}}{\kappa_{p}}%
}\left(  X_{s}sv_{c}^{s+\frac{q-p}{2}-1}+L_{H}\omega_{k}kv_{c}^{k+\frac
{q-p}{2}-1}\right)  .
\end{equation}

Comparing this with the expansion (\ref{a_expan_es}), we see that the dominant
term is obtained by taking $\frac{P}{\sqrt{\mathcal{X}^{2}-N^{2}}}=1$ that
gives us the same expression as in the case of a near-subcritical particle.
So, we see that for all types of particles, the quantity $X_{s}$ is given by
(\ref{xs_nsc_s<k}) if $s<k$ and by (\ref{xs_nsc_s=k}) if $s=k$.

We observe that for all types of particles, $X_{s}$ is only controllable in
the case of $s\leq k$ (for any range of point of collision). We have already
extensively worked on finding the conditions at which we obtain the 1st
solution (that corresponds to $s<k$) or the 2nd solution (which corresponds to
$s\geq k$) in (\ref{s_on_n0}). The corresponding results are given in Tables
\ref{tab_1},\ref{tab_2} and \ref{tab_3}. These tables are useful in the sense
that, for a given type of particle and a given value of $k$, we can easily
deduce which solution for $s$ we can use among the ones given in
(\ref{s_on_n0}). If it is the first solution, then $X_{s}$ is given by
(\ref{xs_nsc_s<k}). However, if it is the second solution in (\ref{s_on_n0})
and additionally $s=k$, then $X_{s}$ is given by (\ref{xs_nsc_s=k}). If it is
the second solution in (\ref{s_on_n0}) and $s>k$, then $X_{s}$ cannot be controlled.

Now, let us discuss how we can manipulate by the particle parameters to make
it possible to reach infinity.

If the particle is near-ultracritical, it may potentially reach infinity if
$\delta>0$ and $X_{s}>0$ (see (\ref{vt_nuc})). These conditions are the same
as for the case of near-subcritical particles that has already been investigated.

If the particle is near-overcritical, it is impossible to make any such
particle to reach infinity (see (\ref{vc_overcr})).

However, if the particle is near-critical, the corresponding particle may
reach infinity only if $\delta>0$ and $A_{p/2}<X_{p/2}$ (see 3-rd condition in
(\ref{vt_nc})). If $s<k$, using (\ref{xs_nsc_s<k}), this gives us:
$(a_{f}^{(r)})_{n_{1}}<\sigma\sqrt{\frac{A_{q}}{\kappa_{p}}}sA_{p/2}$. If
$s=k$, using (\ref{xs_nsc_s=k}), this gives: $(a_{f}^{(r)})_{n_{1}}%
<\sigma\sqrt{\frac{A_{q}}{\kappa_{p}}}(sA_{p/2}+L_{H}\omega_{k}k)$. We can
summarize all these cases in the Table \ref{tab_4}:

\begin{table}[ptb]%
\begin{tabular}
[c]{|c|c|c|c|}\hline
& $s<k$ & $s=k$ & $s>k$\\\hline
NSC & $(a_{f}^{(r)})_{n_{1}}<0$ & $(a_{f}^{(r)})_{n_{1}}<-\sqrt{\frac{A_{q}%
}{\kappa_{p}}}L_{H}\omega_{k}k$ & Impossible\\\hline
NC & $(a_{f}^{(r)})_{n_{1}}<-\sqrt{\frac{A_{q}}{\kappa_{p}}}sA_{p/2}$ &
$(a_{f}^{(r)})_{n_{1}}<-\sqrt{\frac{A_{q}}{\kappa_{p}}}(sA_{p/2}+L_{H}%
\omega_{k}k)$ & Impossible\\\hline
NUC & $(a_{f}^{(r)})_{n_{1}}<0$ & $(a_{f}^{(r)})_{n_{1}}<-\sqrt{\frac{A_{q}%
}{\kappa_{p}}}L_{H}\omega_{k}k$ & Impossible\\\hline
NOC & Impossible & Impossible & Impossible\\\hline
\end{tabular}
\caption{Conditions which have to hold for an acceleration to make different
particles possible to achieve infinity}%
\label{tab_4}%
\end{table}

\section{Summary and Conclusions \label{summary}}

Thus, we have shown that the BSW effect is possible for a quite rich family of
configurations that include the combination of a type of horizon, that of a
near-fine-tuned particle, and a force. We have presented the classification of
different types of near-fine-tuned particles, generalizing the one presented
for fine-tuned ones (near-subcritical, near-critical, and near-ultracritical)
and adding a new type that is possible only for the case of near-fine-tuned
ones - near-overcritical particles (see Table \ref{class_tab}). Particles of
each type differ in the behavior of the components of the four-velocity near
the horizon that causes different kinematical properties. Specifically, we
have analyzed the allowed ranges of motion for each type of particle (and we
have shown that the corresponding ranges of allowed motion are different) and
the near-horizon limits of the components of the four-velocity, that also turn
out to be different. These results have been used to describe the behavior of
energy in the center of mass frame of two colliding particles. It is important
to note that the investigation of near-fine-tuned particles has opened up a
wider variety of different scenarios for particles collision. In this process,
one near-fine-tuned particle may participate with a fine-tuned one (or a usual
one), or there may be two near-fine-tuned particles. For all of these cases,
we have formulated how the energy of collision would behave as the point of
collision approaches the horizon and have formulated the conditions that must
be met to make the energy divergent (which is the main property of the BSW
effect). The corresponding results have been briefly summarized in Tables
\ref{gamma_delta_tab}-\ref{gamma_d_tab}. Table \ref{gamma_delta_tab}
essentially generalizes the situation considered earlier for non-extremal
black holes \cite{gr-pav}, \cite{prd}.

Furthermore, we have focused on the dynamic properties of particles
participating in the BSW phenomenon. We have analyzed the behavior of the
forces acting on such particles and investigated under which conditions these
forces are finite. The corresponding results are briefly summarized in Tables
\ref{tab_1},\ref{tab_2}, and \ref{tab_3}. In Section \ref{sec_dif_types} and
Table \ref{tab_6} we indicated which types of particles experience finite
force for a given type of horizon. Then we focused on the possibility of
preservation kinematic censorship principle in the case of near-fine-tuned
particles, and we showed that this principle holds for them. An additional
investigation concerns the possibility of changing the ranges of motion of
particles by the action of an external force. This is an important topic in
the analysis of the possibility of having the BSW effect and allowing
particles falling from infinity to achieve the horizon and participate in this
phenomenon. The corresponding results have been briefly summarized in Table
\ref{tab_4}.

As the summary of our work, we present Table \ref{fin_tab} and Table
\ref{fin_tab_2} that show all possible cases when the BSW phenomenon is
possible for near-fine-tuned and fine-tuned particles experiencing an action
of a finite force. These tables, being the union of Table \ref{tab_6} (for
near-fine-tuned ones) and Table VIII from \cite{ov-zas-prd} (for fine-tuned
ones) with Tables \ref{gamma_delta_tab}-\ref{gamma_d_tab} and Table II from
\cite{ov-zas-prd} give the final answer to the question related to the
possibility of having BSW phenomenon with finite forces for near-fine-tuned
and fine-tuned particles. Note an important difference between these tables.
In Table \ref{fin_tab_2} both particles have either $v_{c}\ll v_{e,t}$ or
$v_{c}\gg v_{e,t}.$ This effectively means that such particles behave as usual
(in $v_{c}\ll v_{e,t}$ range) or as fine-tuned ones (in $v_{c}\gg v_{e,t}$
range) (see Sections \ref{sec_4_a} and \ref{iv_c}), so this table effectively
describes the possibility of standard BSW phenomenon. However, one can see
that if one of particles is near-fine-tuned and $v_{c}\sim v_{e,t}$ , we
obtain new scenarios described in Table \ref{fin_tab}. The most obvious case
is that appearance of near-fine-tuned particles made the BSW phenomenon
possible, forces for the non-extremal horizon being finite (that is forbidden
for fine-tuned ones). In other words, in Table IX one of particles has
properties specific for near-fine-tuned particles that happens for $v_{c}\sim
v_{e,t}$ . In Table X any particle behaves either as usual or a fine-tuned one.

One reservation is in order. In the case of the "pure" BSW effect NOC
particles do not exist. However, we can consider them formally if $\delta$ is
large enough. For $\delta>>v_{c}^{p/2}$ such a particle is indistinguishable
from a usual one. This justifies why we mention NOC for particle 1 in the 1st
line of Table \ref{fin_tab} and in \ref{fin_tab_2}.

\begin{table}[ptb]%
\begin{tabular}
[c]{|c|c|c|c|c|c|c|c|}\hline
Type of horizon & \multicolumn{3}{|c}{1-st particle's type and range} &
\multicolumn{3}{|c|}{2-nd particle's type and range} & \\\hline
Non-extremal & NOC & $v_{c}\ll v_{t}^{(1)}$ & $\delta_{1}\gg v_{c}^{p/2}$ &
NOC & $v_{c}\sim v_{t}^{(2)}$ & $\delta_{2}\sim v_{c}^{p/2}$ &
(89)\\\cline{2-4}
& \multicolumn{3}{|c|}{U} &  &  &  & \\\hline
Extremal & NC, NUC, NOC & $v_{c}\ll v_{e,t}^{(1)}$ & $\delta_{1}\gg
v_{c}^{p/2}$ & NC, NUC, NOC & $v_{c}\sim v_{e,t}^{(2)}$ & $\delta_{2}\sim
v_{c}^{p/2}$ & (89)\\\cline{2-4}
& \multicolumn{3}{|c|}{U} &  &  &  & \\\hline
Ultraextremal & NSC & $v_{c}\ll v_{e,t}^{(1)}$ & $\delta_{1}\gg v_{c}^{s}$ &
NSC & $v_{c}\sim v_{e,t}^{(2)}$ & $\delta_{2}\sim v_{c}^{s}$ &
(83)\\\cline{2-4}
& \multicolumn{3}{|c|}{U} &  &  &  & \\\cline{2-2}\cline{2-8}\cline{5-5}
& NSC & $v_{c}\sim v_{e,t}^{(1)}$ & $\delta_{1}\sim v_{c}^{s}$ & NSC &
$v_{c}\sim v_{e,t}^{(2)}$ & $\delta_{2}\sim v_{c}^{s}$ & (91)\\\cline{2-8}%
\cline{6-8}
& NSC & $v_{c}\gg v_{e,t}^{(1)}$ & $\delta_{1}\ll v_{c}^{s}$ & NSC &
$v_{c}\sim v_{e,t}^{(2)}$ & $\delta_{2}\sim v_{c}^{s}$ & (83)\\\cline{2-4}
& \multicolumn{3}{|c|}{SC} &  &  &  & \\\cline{2-4}\cline{2-8}\cline{5-8}
& NSC & $v_{c}\ll v_{e,t}^{(1)}$ & $\delta_{1}\gg v_{c}^{s}$ & NC, NUC, NOC &
$v_{c}\sim v_{e,t}^{(2)}$ & $\delta_{2}\sim v_{c}^{p/2}$ & (89)\\\cline{2-4}
& \multicolumn{3}{|c|}{U} &  &  &  & \\\cline{2-8}
& NSC & $v_{c}\sim v_{e,t}^{(1)}$ & $\delta_{1}\sim v_{c}^{s}$ & NC, NUC,
NOC & $v_{c}\sim v_{e,t}^{(2)}$ & $\delta_{2}\sim v_{c}^{p/2}$ &
(92)\\\cline{2-8}
& NSC & $v_{c}\sim v_{e,t}^{(1)}$ & $\delta_{1}\sim v_{c}^{s}$ & NC, NUC &
$v_{c}\gg v_{e,t}^{(2)}$ & $\delta_{2}\ll v_{c}^{p/2}$ & (85)\\\cline{5-7}
&  &  &  & \multicolumn{3}{|c|}{C, UC} & \\\cline{2-8}\cline{5-7}
& NSC & $v_{c}\gg v_{e,t}^{(1)}$ & $\delta_{1}\ll v_{c}^{s}$ & NC, NUC, NOC &
$v_{c}\sim v_{e,t}^{(2)}$ & $\delta_{2}\sim v_{c}^{p/2}$ & (89)\\\cline{2-4}
& \multicolumn{3}{|c|}{SC} &  &  &  & \\\cline{2-8}\cline{2-8}
& NC, NUC, NOC & $v_{c}\ll v_{e,t}^{(1)}$ & $\delta_{1}\gg v_{c}^{p/2}$ & NC,
NUC, NOC & $v_{c}\sim v_{e,t}^{(2)}$ & $\delta_{2}\sim v_{c}^{p/2}$ &
(89)\\\cline{2-4}\cline{2-4}
& \multicolumn{3}{|c|}{U} &  &  &  & \\\hline
\end{tabular}
\caption{Table showing for which particles and which ranges of their motion
the BSW phenomenon is possible if forces acting on both particles are finite.
This table describes all new possible cases when both particles are
near-fine-tuned or one is fine-tuned while the second particle is
near-fine-tuned. In near-horizon collisions near-fine-tuned particles with
$v_{c}\ll v_{t,e}$ behave similarly to usual ones. If $v_{c}\gg v_{t,e}$,
near-fined particles are similar to fine-tuned ones. The last column displays
equation number describing a corresponding case.}%
\label{fin_tab}%
\end{table}

\begin{table}[ptb]%
\begin{tabular}
[c]{|c|c|c|c|c|c|c|c|}\hline
Type of horizon & \multicolumn{3}{|c|}{1-st particle's type and range} &
\multicolumn{3}{|c|}{2-nd particles's type and range} & \\\hline
Extremal & NC, NUC, NOC & $v_{c}\ll v_{e,t}^{(1)}$ & $\delta_{1}\gg
v_{c}^{p/2}$ & NC, NUC & $v_{c}\gg v_{e,t}^{(2)}$ & $\delta_{2}\ll v_{c}%
^{p/2}$ & 4\\\cline{2-4}
& \multicolumn{3}{|c|}{U} &  &  &  & \\\cline{2-7}
& \multicolumn{3}{|c|}{U} & \multicolumn{3}{|c|}{C, UC} & \\\hline
Ultraextremal & NSC, NC, NUC, NOC & $v_{c}\ll v_{e,t}^{(1)}$ & $\delta_{1}\gg
v_{c}^{p/2}$ & NSC & $v_{c}\gg v_{e,t}^{(2)}$ & $\delta_{2}\ll v_{c}^{s}$ &
2\\\cline{2-4}
& \multicolumn{3}{|c|}{U} &  &  &  & \\\cline{2-7}
& \multicolumn{3}{|c|}{U} & \multicolumn{3}{|c|}{SC} & \\\cline{2-8}
& NSC, NC, NUC, NOC & $v_{c}\gg v_{e,t}^{(1)}$ & $\delta_{1}\ll v_{c}^{p/2}$ &
NSC & $v_{c}\gg v_{e,t}^{(2)}$ & $\delta_{2}\ll v_{c}^{s}$ & 3\\\cline{2-4}
& \multicolumn{3}{|c|}{SC} &  &  &  & \\\cline{2-7}
& \multicolumn{3}{|c|}{SC} & \multicolumn{3}{|c|}{SC} & \\\cline{2-8}
& NSC & $v_{c}\ll v_{e,t}^{(1)}$ & $\delta_{1}\gg v_{c}^{s}$ & NC, NUC &
$v_{c}\gg v_{e,t}^{(2)}$ & $\delta_{2}\ll v_{c}^{p/2}$ & 4\\\cline{2-4}
& \multicolumn{3}{|c|}{U} &  &  &  & \\\cline{2-7}
& \multicolumn{3}{|c|}{U} & \multicolumn{3}{|c|}{C, UC} & \\\cline{2-8}
& NSC & $v_{c}\gg v_{e,t}^{(1)}$ & $\delta_{1}\ll v_{c}^{s}$ & NC, NUC &
$v_{c}\gg v_{e,t}^{(2)}$ & $\delta_{2}\ll v_{c}^{p/2}$ & 4\\\cline{2-4}
& \multicolumn{3}{|c|}{SC} &  &  &  & \\\cline{2-7}
& \multicolumn{3}{|c|}{SC} & \multicolumn{3}{|c|}{C, UC} & \\\hline
\end{tabular}
\caption{Table showing for which particles and which ranges of their motion
BSW phenomenon is possible if forces acting on both particles are finite. This
table describes all the cases with the BSW effect when usual, fine-tuned or
near-fined particles participate. The last column indicates the line number in
Table II from \cite{ov-zas-prd} that describes a corresponding case.}%
\label{fin_tab_2}%
\end{table}

\appendix

\section{Computation of acceleration in FZAMO frame\label{app_acel}}

In this investigation, we are dealing with near-fine-tuned particles that have
a non-zero (but quite small) $P$ on a horizon. As a result, the radial
velocity for such particles is non-zero on a horizon, and to compute the
components of acceleration, we need to choose a suitable frame. We will use a
frame called the FZAMO (Frame of Zero Angular Momentum Observer), which is
attached to a falling particle. This frame is such that the three-velocity of
a particle is zero in this frame. To obtain this frame, we start with a
stationary tetrad (\ref{tetr_1}-\ref{tetr_2}). The three-velocity of a
particle in this frame is given by:%

\begin{equation}
V^{(i)}=-\frac{e_{\mu}^{(i)}u^{\mu}}{e_{\mu}^{(0)}u^{\mu}}=(\frac{\sigma
P}{\mathcal{X}},0,\frac{\mathcal{L}N}{\mathcal{X}\sqrt{g_{\varphi\varphi}}%
})=|v|(\cos\psi,0,\sin\psi),
\end{equation}

where%
\begin{equation}
|v|=\sqrt{1-\frac{N^{2}}{\mathcal{X}^{2}}},\text{ \ \ }\tan\psi=\frac
{N\mathcal{L}}{P\sqrt{g_{\varphi\varphi}}}.
\end{equation}

The FZAMO frame may be obtained if we perform several transformations to a
stationary tetrad:

\begin{itemize}
\item Rotate frame in such a way that new radial tetrad vector is co-directed
with direction of 3-velocity:%
\begin{align}
\widetilde{e}_{\mu}^{(0)}  &  =e_{\mu}^{(0)},\text{ }\ \text{\ }\widetilde
{e}_{\mu}^{(2)}=e_{\mu}^{(2)},\\
\widetilde{e}_{\mu}^{(1)}  &  =e_{\mu}^{(1)}\cos\psi+e_{\mu}^{(3)}\sin
\psi,\text{ \ \ }\widetilde{e}_{\mu}^{(3)}=e_{\mu}^{(3)}\cos\psi-e_{\mu}%
^{(1)}\sin\psi.
\end{align}

\item Perform a boost in a direction of particle's motion:%
\begin{align}
e_{\mu}^{(2)\prime}  &  =\widetilde{e}_{\mu}^{(2)},\text{ }\ \text{\ }e_{\mu
}^{(3)\prime}=\widetilde{e}_{\mu}^{(3)},\\
e_{\mu}^{(0)\prime}  &  =\gamma(\widetilde{e}_{\mu}^{(0)}-|v|\widetilde
{e}_{\mu}^{(1)}),\text{ \ \ }e_{\mu}^{(1)\prime}=\gamma(\widetilde{e}_{\mu
}^{(1)}-|v|\widetilde{e}_{\mu}^{(0)}).
\end{align}

where $\gamma=\frac{1}{\sqrt{1-v^{2}}}=\frac{\mathcal{X}}{N}.$
\end{itemize}

After these actions we obtain tetrad vectors in a form:%

\begin{align}
e_{\mu}^{(0)\prime}  &  =(\mathcal{X}+\omega\mathcal{L},-\frac{P}{\sqrt{A}%
N},0,-\mathcal{L)},\\
e_{\mu}^{(1)\prime}  &  =\frac{1}{\sqrt{\mathcal{X}^{2}-N^{2}}}\left(
-\omega\mathcal{LX-X}^{2}+N^{2},\frac{P\mathcal{X}}{\sqrt{A}N},0,\mathcal{LX}%
\right)  ,\\
e_{\mu}^{(2)\prime}  &  =\sqrt{g_{\theta\theta}}(0,0,1,0),\\
e_{\mu}^{(3)\prime}  &  =\frac{1}{\sqrt{\mathcal{X}^{2}-N^{2}}}\left(  -\omega
P\sqrt{g_{\varphi\varphi}},-\frac{\mathcal{L}N}{\sqrt{g_{\varphi\varphi}A}%
},0,\sqrt{g_{\varphi\varphi}}P\right)  .
\end{align}

One can easily check than indeed in this frame corresponding 3-velocity%
\begin{equation}
V^{(i)^{\prime}}=-\frac{e_{\mu}^{(i)\prime}u^{\mu}}{e_{\mu}^{(0)\prime}u^{\mu
}},
\end{equation}
is zero.

Acceleration components in this frame may be computed from the definition
$a_{F}^{(a)}=e_{\mu}^{(a)\prime}a^{\mu}.$ Using expressions for acceleration
in an OZAMO frame, given in Section VI in \cite{ov-zas-prd}%
\begin{align}
a^{t}  &  =\frac{a_{o}^{(t)}}{N}=\frac{u^{r}}{N^{2}}(\partial_{r}%
\mathcal{X+L}\partial_{r}\omega),\\
a^{r}  &  =\sqrt{A}a_{o}^{(r)}=\mathcal{X}\frac{A}{N^{2}}\left(  \partial
_{r}\mathcal{X+L}\partial_{r}\omega-\frac{N^{2}}{\mathcal{X}}\frac
{\mathcal{L}\partial_{r}\mathcal{L}}{g_{\varphi\varphi}}\right)  ,\\
a^{\varphi}  &  =\frac{a_{o}^{(\varphi)}}{\sqrt{g_{\varphi\varphi}}}+\omega
a^{t}=u^{r}\left(  \frac{\partial_{r}\mathcal{L}}{g_{\varphi\varphi}}%
+\frac{\omega}{N^{2}}(\partial_{r}\mathcal{X+L}\partial_{r}\omega)\right)  ,
\end{align}
we have%
\begin{align}
a_{F}^{(t)}  &  =0,\label{a_f_t}\\
a_{F}^{(r)}  &  =\frac{u^{r}}{\sqrt{\mathcal{X}^{2}-N^{2}}}(\partial
_{r}\mathcal{X+L}\partial_{r}\omega),\label{a_f_r}\\
a_{F}^{(\varphi)}  &  =\frac{1}{\sqrt{\mathcal{X}^{2}-N^{2}}}\frac{\sqrt{A}%
}{N}\left[  (\mathcal{X}^{2}-N^{2})\frac{\partial_{r}\mathcal{L}}%
{\sqrt{g_{\varphi\varphi}}}-\frac{\mathcal{LX}}{\sqrt{g_{\varphi\varphi}}%
}(\partial_{r}\mathcal{X+L}\partial_{r}\omega)\right]  . \label{a_f_phi}%
\end{align}

\section{Proper time\label{sec_prop_time}}

In this section we are going to analyze the proper time of near-fine-tuned
particles for different scenarios of motion of particle. Our main aim is to
find out under which conditions the proper time will be finite and how its
behaviour for near-fine-tuned particle correlates with corresponding
fine-tuned ones. To this end, let us consider separately different ranges of
particle motion.

\subsection{$v_{c}\gg v_{e,t}$}

In this case, as was shown in Section \ref{sec_4_a}, behaviour of all physical
quantities is the same as for fine-tuned particles of the same type. Thus the
proper time given by%
\begin{equation}
\tau=\int\frac{dr}{u^{r}}+C,
\end{equation}
in this range of coordinates can be found in \cite{ov-zas-prd}. Here $C$ is a
constant of integration, which we will omit in a further analysis because we
are interested in a near-horizon behaviour of the proper time, which is
independent on $C.$

\subsection{$v_{c}\ll v_{e,t}$}

In this case we can refer to the Section \ref{iv_c} and simply take the limit
$v_{c}\longrightarrow0$ while keeping terms with $\delta$ in expressions for
the $P$ and $\mathcal{X}$. Thus we can write:%
\begin{equation}
P=\sqrt{\mathcal{X}^{2}-\left(  1+\frac{L^{2}}{g_{\varphi\varphi}}\right)
N^{2}}\approx\delta.
\end{equation}

The radial velocity in this case is given by%
\begin{equation}
u^{r}=\frac{\sqrt{A}}{N}P\approx\sqrt{\frac{A_{q}}{\kappa_{p}}}v_{c}%
^{\frac{q-p}{2}}\delta.
\end{equation}

The proper time is given by%

\begin{equation}
\tau=\int\frac{dr}{u^{r}}\approx\sqrt{\frac{\kappa_{p}}{A_{q}}}\frac
{v_{c}^{\frac{p-q}{2}+1}}{\delta(\frac{p-q}{2}+1)}.
\end{equation}

We see that in this range of coordinates $\tau\sim v_{c}^{\frac{p-q}{2}+1}$
and this result does not depend on a type of corresponding near-fine-tuned
particle. Comparing this with Table 1 in \cite{ov-zas-prd}, we see, that the
proper time in this range is the same as for usual particles. This is not
surprising, because, as was discussed in Section \ref{iv_c}, in range
$v_{c}\ll v_{e,t}$ all kinematic properties of all near-fine-tuned particles
are the same as for usual ones. The proper time is finite if $q<p+2.$

\subsection{$v_{c}\sim v_{e,t}$}

In this range of particle's motion expression for proper time cannot be
explicitly integrated. However we are rather interested in an asymptotical
behaviour of the proper time then in an exact expression. For this we can use
that in $v_{c}\sim v_{t}$ range holds $P\sim\delta$ (see Section \ref{iv_b}).
This allows us to write:%

\begin{equation}
\tau=\int\frac{dr}{u^{r}}\approx\sqrt{\frac{\kappa_{p}}{A_{q}}}\int\frac
{dr}{v^{\frac{q-p}{2}}P}\sim\frac{1}{\delta}\int\frac{dr}{v^{\frac{q-p}{2}}%
}\sim\frac{v_{c}^{\frac{p-q}{2}+1}}{\delta}.
\end{equation}

Note that this expression depends on both $v_{c}$ and $\delta.$ However we can
relate these quantities as far as we consider $v_{c}\sim v_{e,t}$ case. For
near-subcritical particles one can use that $v_{e,t}\sim\delta^{1/s}$
according to (\ref{vt_subcr}), (\ref{vc_subcr}), or, inverting $\delta\sim
v_{c}^{s},$ we have:%
\begin{equation}
\tau\sim v_{c}^{\frac{p-q}{2}+1-s}.
\end{equation}

As for near-subcritical particles $0<s<p/2,$ we see that proper time behaves
as $\tau\sim v_{c}^{-\alpha},$ where $\frac{q-p-2}{2}<\alpha<\frac{q-2}{2}.$
Comparing this result with Table 1 in \cite{ov-zas-prd}, we see that the
proper time for near-subcritical particles in region \ $v_{c}\sim v_{e,t}$
behaves in the same way as for corresponding subcritical particles. The proper
time for such particles is finite if $s<\frac{p-q}{2}+1.$

For near-critical, near-ultracritical and near-overcritical ones situation
differs. For them holds $\delta\sim v_{c}^{p/2},$ what allows us to write:%
\begin{equation}
\tau\sim v_{c}^{\frac{2-q}{2}}.
\end{equation}

From this we see that $\tau\sim v_{c}^{-\alpha}$ with $\alpha=\frac{q-2}{2}.$
Comparing this with Table I in \cite{ov-zas-prd}, we see, that in range
$v_{c}\sim v_{t}$ behaviour of the proper time for near-critical particles is
the same as for critical ones. However, for near-ultracritical behaviour is
not the same as for ultracritical ones and is the same as for critical ones
(this also concerns near-overcritical ones). The proper time is finite if
$q<2.$

\end{document}